\documentclass[12pt]{article}
\usepackage[pdftex,bookmarksopen=true,bookmarks=true,unicode,setpagesize]{hyperref}
\hypersetup{colorlinks=true,linkcolor=black,citecolor=black}

\usepackage{amssymb, amsmath, amsthm}

\usepackage{appendix}

\numberwithin{equation}{section}

\usepackage{xcolor}

\allowdisplaybreaks
\usepackage{cite}

\newcommand\R{\mathbb{R}}

\newcommand{\vertiii}[1]{{\vert\kern-0.25ex\vert\kern-0.25ex\vert #1
    \vert\kern-0.25ex\vert\kern-0.25ex\vert}}

\newtheorem{theorem}{Theorem}[section]
\newtheorem{corollary}[theorem]{Corollary}
\newtheorem{lemma}[theorem]{Lemma}
\newtheorem{proposition}[theorem]{Proposition}

\theoremstyle{remark}

\theoremstyle{remark}

\theoremstyle{remark}
\newtheorem{remark}[theorem]{Remark}

\begin{document}

\vspace{-20mm}
\begin{center}{\Large \bf 
Generalized  Segal--Bargmann transforms and generalized Weyl algebras associated with the~Meixner class of orthogonal polynomials}
\end{center}

{\large Chadaphorn Kodsueb}\footnote{The current affiliation: School of Mathematical Sciences and Geoinformatics, Institute of Science, Suranaree University of Technology, Nakhon Ratchasima, Thailand 30000}\\ Department of Mathematics, Swansea University, Bay Campus,  Swansea SA1 8EN, U.K.;
e-mail: \texttt{c.punkumkerd@swansea.ac.uk}\vspace{2mm}

{\large Eugene Lytvynov}\\ Department of Mathematics, Swansea University, Bay Campus,  Swansea SA1 8EN, U.K.;
e-mail: \texttt{e.lytvynov@swansea.ac.uk} (Corresponding Author)\vspace{2mm}

{\small
\begin{center}
{\bf Abstract}
 \end{center}
 
\noindent  Meixner (1934) proved that there exist exactly five classes of orthogonal Sheffer sequences: Hermite polynomials which are orthogonal with respect to  Gaussian distribution, Charlier polynomials orthogonal with respect to Poisson distribution, Laguerre polynomials orthogonal with respect to gamma distribution, Meixner polynomials of the first kind, orthogonal with respect to negative binomial  distribution, and Meixner polynomials of the second kind,  orthogonal with respect to Meixner distribution. The Segal--Bargmann transform provides a unitary isomorphism between the $L^2$-space of the Gaussian distribution and the Fock or Segal--Bargmann  space of entire funcitons. This construction was also extended to the case of the Poisson distribution. The present paper deals with the latter three classes of orthogonal Sheffer sequences.  By using a set of nonlinear coherent states, we construct and study a generalized Segal--Bargmann transform which is a unitary isomorphism between the $L^2$-space of the orthogonality measure and a certain Fock space of entire functions.  To derive  our results, we use normal ordering in  generalized Weyl algebras that are naturally associated with the orthogonal Sheffer sequences.

 \noindent

 } \vspace{2mm}



\section{Introduction}

Fock spaces play a fundamental role in quantum
mechanics as well as in infinite-dimensional analysis and probability, both classical and noncommutative (quantum), see e.g.\ \cite{Derezinski,Meyer,Parthasarathy}. Roughly speaking, a symmetric Fock space is an infinite orthogonal sum of symmetric $n$-particle Hilbert spaces. There exists an alternative description of a symmetric Fock space as a space of holomorphic functions. Such a space is usually called the Segal--Bargmann space.

	Let us briefly discuss the Segal--Bargmann construction in the one-dimensional case. Bargmann~\cite{s-Bargmann} defined a Hilbert space $\mathbb F (\mathbb C)$ as the closure of polynomials over~$\mathbb C$ in the $L^2$-space $L^2 (\mathbb C, \nu)$. Here  $\nu$ is the Gaussian measure on~$\mathbb C$ given by $\nu (dz) = \pi^{-1} \exp (- \lvert z \rvert^2) \, dA(z)$,  where $dA(z)$ is the Lebesgue measure on~$\mathbb C$. The monomials $(z^n)_{n=0}^\infty$ form an orthogonal basis for $\mathbb F (\mathbb C)$ with $(z^m,z^n)_{\mathbb F(\mathbb C)}=n! \, \delta_{m,n}$. Here and below, $\delta_{m,n}$ denotes the Kronecker delta. 	The $\mathbb F (\mathbb C)$ consists of entire functions $\varphi(z) = \sum_{n=0}^\infty f_n  z^n$ that satisfy $\sum_{n=0}^\infty  \lvert f_n \rvert^2 \, n! < \infty$. The $\mathbb F (\mathbb C)$  is a reproducing kernel Hilbert space with reproducing kernel $\mathbb K (z,w) = \sum_{n=0}^\infty (n!)^{-1} \, (\bar{z} w)^n$.
	
	Let $\mu$ be the standard Gaussian distribution on $\mathbb R$ and let $(h_n)_{n=0}^\infty$ be the sequence of monic Hermite polynomials that form an orthogonal basis for $L^2 (\mathbb R, \mu)$. The Segal--Bargmann transform is the unitary operator $\, \mathbb S : L^2 (\mathbb R, \mu) \to \mathbb F (\mathbb C)$ that satisfies $(\mathbb S \, h_n)(z) = z^n$. This operator has a representation through the coherent states:
	$$\mathbb E (x, z) = \sum_{n=0}^\infty \, \frac{z^n }{n!} \,  h_n (x) = \exp \left( -\frac{1}{2}(z^2 - 2 x z) \right),\quad x\in\mathbb R,\ z\in\mathbb C.$$ 
	More precisely, for $f\in L^2(\mathbb R,\mu)$ and $z\in\mathbb C$, one has  
	$(\mathbb S f)(z) = \int_{\mathbb R} f(x) \mathbb E (x, z) \, \mu(dx)$.  
	For a fixed $z\in\mathbb C$, $\mathbb E(\cdot,z)$ is an eigenfunction of the lowering operator in $L^2 (\mathbb R, \mu)$ with eigenvalue $z$. More exactly, if we define the (unbounded) lowering operator $\partial^-$ in $L^2 (\mathbb R, \mu)$ by $\partial^- h_n = n h_{n-1}$, then $\partial^- \, \mathbb E (\cdot, z) = z \, \mathbb E (\cdot, z)$. For  $z$ real, the operator $\mathbb S$ can also be written as 
\begin{equation}\label{fdzraer}
(\mathbb S f)(z) = \int_{\mathbb R} f(x+z)\, \mu(dx),\quad f \in L^2 (\mathbb R, \mu),\ z\in\mathbb R.
\end{equation}

	Let also $\partial^+$ denote the raising operator for the Hermite polynomials: $\partial^+ h_n = h_{n+1}$. Then, the operator of multiplication by the variable $x$ in $L^2 (\mathbb R, \mu)$ has the form $\partial^+ + \partial^-$. Hence, under the Segal--Bargmann transform $\mathbb S$, this operator goes over to the operator $Z+D$, where $Z$ is the multiplication by the variable $z$ in $\mathbb F (\mathbb C)$, and $D$ is the differentiation in $\mathbb F (\mathbb C)$. In this setting, the operators $Z$ and $D$ are adjoint of each other. Note that these operators satisfy the commutation relation $[D,Z] = 1$, hence they are generators of a Weyl algebra, see e.g. \cite[Chapter 5]{MansourSchork}.

The Segal--Bargmann transform for the Gaussian measure admits an extension to both the multivariate case 
\cite{s-Bargmann} and an infinite-dimensional case, see e.g.\ \cite{GrossMalliavin} and  \cite[Section 3.3]{Obata}.

	Asai et al.~\cite{Segal-Bargmann} constructed a counterpart of the Segal--Bargmann transform in the case of the Poisson distribution with parameter $\sigma > 0$: 
$\pi_\sigma (d\xi) = e^{-\sigma} \sum_{n=0}^\infty  \frac1{n!}\, \sigma^n\, \delta_n (d\xi)$ ($\delta_n$ denoting the Dirac measure at $n$).
Define the Gaussian measure $\nu_\sigma$ on $\mathbb C$ by
\begin{equation}\label{zsawsa}
	 \nu_\sigma (dz) = \frac{1}{\pi\sigma} \, \exp{\bigg( - \frac{|z|^2}{\sigma} \bigg)} \ dA(z).
	 \end{equation}
Let  the Hilbert space $\mathbb F_\sigma (\mathbb C)$ be the closure of polynomials over~$\mathbb C$ in $L^2 (\mathbb C, \nu_\sigma)$. The monomials $(z^n)_{n=0}^\infty$ form an orthogonal basis for $\mathbb F _\sigma(\mathbb C)$ with $(z^m,z^n)_{\mathbb F_\sigma(\mathbb C)}=\sigma^n\,n! \, \delta_{n,m}$. The $\mathbb F_\sigma (\mathbb C)$ consists of entire functions $\varphi(z) = \sum_{n=0}^\infty f_n \, z^n$ that satisfy $\sum_{n=0}^\infty  \lvert f_n \rvert^2 \, \sigma^n\,n! < \infty$. 
Let $(c_n)_{n=0}^\infty$ be the sequence of monic Charlier polynomials  that form an orthogonal basis for $L^2 (\mathbb N_0, \pi_\sigma)$ (here and below we denote $\mathbb N_0=\{0,1,2,\dots\}$).
 The generalized Segal--Bargmann transform is a unitary operator $\mathbb S : L^2 (\mathbb N_0, \pi_\sigma) \to {\mathbb F}_\sigma (\mathbb C)$ satisfying $(\mathbb S c_n)(z) = z^n$. The corresponding coherent states are\footnote{In this paper, we always denote a (generalized) Segal--Bargmann trasnform by $\mathbb S$ and the corresponding coherent states by $\mathbb E(\cdot,\cdot)$. This should not lead to a confusion, since it will always be clear from the context which particular choice of the distribution on $\R$  we are dealing with.} 
 $${\mathbb E} (\xi, z) = \sum_{n=0}^\infty \frac{z^n}{n!\,\sigma^n}c_n(\xi)=e^{-z} \left( 1+ \frac{z}{\sigma} \right)^\xi,\quad\xi\in\mathbb N_0,\ z\in\mathbb C.$$
  It holds that $\sigma\partial^- \, {\mathbb E}(\cdot, z) = z \, {\mathbb E}(\cdot, z)$, where $\partial^-$ is the lowering operator for the Charlier polynomials $(c_n)_{n=0}^\infty$. Note that $\sigma\partial^-$ is the adjoint of the raising operator $\, \partial^+$ for the polynomials $(c_n)_{n=0}^\infty$.

A key difference with the Gaussian case is that, under the transformation $\mathbb S$, the operator of multiplication by the variable $\xi$ goes over to the operator $\rho=\mathcal U  \mathcal V$ in $\mathbb F_\sigma(\mathbb C)$, where 
		\begin{equation} 	\label{mn67-0tr4}
			\mathcal U = Z + \sigma, \quad \mathcal V = D + 1.
		\end{equation}
	Note that the operators $\mathcal U$ and $\mathcal V$ still satisfy the commutation relation $[\mathcal V, \mathcal U] =1$, hence $\mathcal U$ and $\mathcal V$ generate a Weyl algebra.
			
Both Hermite polynomials $(h_n)_{n=0}^\infty$ and Charlier polynomials $(c_n)_{n=0}^\infty$
belong to the class of orthogonal Sheffer sequences. Recall that a monic polynomial sequence $(s_n)_{n=0}^\infty$ over $\R$ is called a Sheffer sequence if its (exponential) generating function is of the form
\begin{equation}\label{vcxresayw}
\sum_{n=0}^\infty\frac{t^n}{n!}\,s_n(x)=\exp\big[A(t)+xB(t)\big],
\end{equation}
where $A(t)$ and $B(t)$ are formal power series over $\R$ satisfying $A(0)=B(0)=0$ and $B'(0)=1$. 	
	
	Meixner \cite{13-Meixner} proved that there exist exactly five classes of orthogonal Sheffer sequences. In fact, a monic polynomial sequence $(s_n)_{n=0}^\infty$ is an orthogonal Sheffer sequence if and only if it satisfies the recurrence relation
		\begin{equation} \label{545-980+polk}
			xs_n (x) = s_ {n+1} (x) + (\lambda n+l) s_n (x) +\big( \sigma n+\eta n (n-1)\big) s_ {n-1}(x),
		\end{equation}
	where $\lambda \in \R$, $l \in \R$, $\sigma > 0$ and $\eta \ge 0$. The transformation of the constants 
$(\lambda,l)\mapsto(-\lambda,-l)$ corresponds to the push-forward of the orthogonality measure under the map $\R\ni x\mapsto -x\in\mathbb R$.  Hence, we may assume that $\lambda\ge0$. 	
	The constant $l$  corresponds to the shift of the orthogonality measure by $l$, so it can be chosen appropriately, depending on the other three constants. It is also convenient to introduce parameters $\alpha, \beta \in \mathbb C$ that satisfy $\alpha + \beta = \lambda$, $\alpha\beta = \eta$. In the case of both Hermite and Charlier polynomials, we have $\eta=0$.
	
	In this paper, we will deal with the case $\eta > 0$, which corresponds to the other three classes of orthogonal Sheffer sequences. More exactly, for $\alpha=\beta>0$ and $l=\sigma/\alpha$, we obtain the sequence of Laguerre polynomials which are orthogonal with respect to the following gamma distribution on $\mathbb R_+=(0,\infty)$:
\begin{equation}\label{buyfr7iei9}
		\mu_{\alpha, \alpha, \sigma} (dx) = 
		 \frac{1}{\Gamma (\frac\sigma\eta)} \, \alpha^{-\frac\sigma\eta}
		 x^{-1+ \frac\sigma\eta} \, e^{- \frac{x}{\alpha}} \, dx.
		\end{equation} 
For $\alpha>\beta>0$ and $l=\sigma/\alpha$, we obtain the sequence of Meixner polynomials of the first kind which are orthogonal with respect to the following negative binomial (Pascal) distribution on $(\alpha-\beta)\mathbb N_0$:	
\begin{equation}\label{buyfbufyt7r7iei9}
\mu_{\alpha, \beta, \sigma} (dx) = \left( 1-\frac{\beta}{\alpha} \right)^{\frac{\sigma}{\eta}} \sum_{n=0}^\infty \left( \frac{\beta}{\alpha} \right)^n \frac{1}{n!} \left( \frac{\sigma}{\eta} \right)^{(n)} \delta_{(\alpha - \beta)n} \,  (dx) .
\end{equation}
Finally, for $\Re(\alpha)\ge0$, $\Im(\alpha)>0$, $\beta=\overline \alpha$ and $l=0$, we obtain the sequence of Meixner polynomials of the second  kind (or Meixner--Pollaczak polynomials) which are orthogonal with respect to the following Meixner distribution on $\mathbb R$: 
\begin{equation}\label{cftstew5}
	 {\mu}_{\alpha, \beta, \sigma}(dx) ={C}_{\alpha, \beta, \sigma} \, \exp\left( \frac{(\frac\pi2 - \operatorname{Arg}(\alpha))x}{\Im(\alpha)} \right) \,
	\bigg\lvert \Gamma \left( \frac{ix}{2\Im(\alpha)} + \frac{i\sigma\beta}{2\eta\Im(\alpha)} \right) \bigg\rvert^2 \, dx,	
\end{equation}
where $\operatorname{Arg}(\alpha) \in (0, \pi/2]$ and the constant ${C}_{\alpha, \beta, \sigma}$ is given by 
\begin{equation}\label{sewa4tq}
	{C}_{\alpha, \beta, \sigma} = \frac{\big(2\cos\big(\frac\pi2 - \operatorname{Arg}(\alpha)\big)\big)^{\frac{\sigma}{\eta }}}{4\Im(\alpha)\pi \, \Gamma(\frac{\sigma}{\eta })} \, \exp\left( \frac{(\frac\pi2-\operatorname{Arg}(\alpha))\sigma\Re(\alpha)}{\Im(\alpha)\eta } \right). 
\end{equation}

	The aim of the paper is to  study a generalized Segal--Bargmann transform which  is a unitary operator $\mathbb S: L^2(\mu_{\alpha,\beta,\sigma})\to \mathbb F_{\eta,\sigma}(\mathbb C)$ satisfying $(\mathbb Ss_n)(z)=z^n$, where $\mathbb F_{\eta,\sigma}(\mathbb C)$ is a Fock space of entire functions to be defined below. This Segal--Bargmann transform has been previously discussed by Feinsilver~\cite{Feinsilver2} and  Asai \cite{Asai1,Asai2} from the viewpoints of orthogonal polynomials, quantum probability, and representation theory. See also \cite{Alpay-8,Alpay,Karp1,Karp2}.
	
	For $h\in\mathbb C$, let $((\cdot 
	\mid h)_n)_{n=0}^\infty$ denote the sequence of generalized factorials with increment $h$ \cite{HsuShiue}, i.e., for $z\in\mathbb C$, $(z \mid h)_0 = 1$ and 
	\begin{equation}\label{cyrtdstevf}
			(z\mid h)_n = z(z- h)(z-2h) \dotsm (z-(n-1)h),\quad n\in\mathbb N.\end{equation}
						In particular, $(z\mid1)_n=(z)_n$ is a falling factorial and $(z\mid{-1})_n=(z)^{(n)}$ is a rising factorial. Note that the so-called $h$-derivative, $(D_h f)(z) =h^{-1}(f(z+h)-f(z))$, is the lowering operator for this polynomial sequence: $(D_h(\cdot\mid h)_n)(z)=n(z\mid h)_{n-1}$.

For $\sigma>0$ and $\eta\ge0$, we define $\mathbb F_{\eta,\sigma}(\mathbb C)$ as the Hilbert space of entire functions 	
 $\varphi(z) = \sum_{n=0}^\infty f_n \, z^n$ that satisfy 
 \begin{equation}\label{vctresay5}
\sum_{n=0}^\infty \lvert f_n \rvert^2 \, n! \, (\sigma \mid -\eta)_{n} < \infty,
\end{equation}
  and $(z^m,z^n)_{\mathbb F_{\eta,\sigma}(\mathbb C)}=\delta_{m,n}(\sigma\mid-\eta)_n\,n!\,$. Note that, for $\eta=0$, we have $(\sigma\mid0)_n=\sigma^n$ and so $\mathbb F_{0,\sigma}(\mathbb C)=\mathbb F_{\sigma}(\mathbb C)$. 
  
  For general $\sigma>0$ and $\eta>0$, we prove that  $\mathbb F_{\eta,\sigma}(\mathbb C)$ is the closure of the polynomials over $\mathbb C$ in the $L^2$-space $L^2(\mathbb C,\lambda_{\eta,\sigma})$. Here $\lambda_{\eta,\sigma}$ is the random Gaussian measure $\nu_r$ (see formula~\eqref{zsawsa}) where the 
 random variable $r$ (the variance of $\nu_r$) is  distributed according to the gamma distribution 
  $\mu_{\eta,\eta, \eta \sigma}$. The $\mathbb F_{\eta, \sigma} (\mathbb C)$ is a reproducing kernel Hilbert space with reproducing kernel 
$ \mathbb K(z,w) = \sum_{n=0}^\infty \frac{(\bar{z} w)^n}{n! \, (\sigma \mid -\eta)_n}$.

We note that Asai \cite{Asai1} derived a representation of the density of the measure $\lambda_{\eta,\sigma}$ which involves the modified Bessel function.  Furthermore, it was shown in  \cite{Asai1} that $\lambda_{\eta,\sigma}$ is the unique probability measure on $\mathbb C$ whose $L^2$-space contains the Hilbert space~$\mathbb F_{\eta,\sigma}(\mathbb C)$ as its subspace.   In the case $\sigma=\eta=1$, the  space $\mathbb F_{1, 1} (\mathbb C)$ was also studied by Alpay et al.~\cite[Section 9]{Alpay-8} and Alpay and Porat \cite{Alpay}, see also \cite{Karp1, Karp2}.

The generalized Segal--Bargmann transform $\mathbb S: L^2(\mu_{\alpha,\beta,\sigma})\to \mathbb F_{\eta,\sigma}(\mathbb C)$  admits a re\-pre\-sentation 	
		$$ (\mathbb S f)(z) = \int_{\mathbb R} f(x) {\mathbb E} (x, z) \, \mu_{\alpha, \beta, \sigma} (dx), $$
	where 	
\begin{equation}\label{nb.mj.09} 
	 {\mathbb E}(x,z)=\sum_{n=0}^\infty \frac{z^n}{n! \, {(\sigma \mid -\eta)}_n} \, s_n (x),
	 \end{equation}
	 and $\mathbb E(\cdot,z)\in L^2(\mu_{\alpha,\beta,\sigma})$ for each $z\in\mathbb C$. 
Hence, $\big({\mathbb E} (\mathbf{\cdot},z)\big)_{z \in \mathbb C}$ are nonlinear coherent states corresponding to the sequence of numbers $\rho_n= n! \, (\sigma \mid -\eta)_n$ ($n\in\mathbb N_0$). See e.g.\ 
 \cite{La-Othmane-23, GazeauKlauder, Sivakumar} for studies of nonlinear coherent states. For applications of (generalized) coherent states in physics, see e.g. \cite{Gazeau, Perelomov}.
 
 In the special case where $\eta = 1$ and $\sigma = 2j$ with $ j \in \{ 1, \frac{1}{2}, 2, \frac{2}{3}, \ldots \} $, we get $ \rho_n = n! \, (2j)^{(n)}$.
Nonlinear coherent states with such a choice of $\rho_n$ are called the Barut--Girardello states~\cite{BG}, see also \cite[Section 1.1.3]{La-Othmane-23}. Such states appeared in \cite{BG} in a study of coherent states associated with the Lie algebra of the group $SU (1,1)$. For the general choice of the parameters $\lambda$, $\eta$ and $\sigma$, Feinsilver~\cite[Sections~1 and~3.8]{Feinsilver2} obtained a representation of the function $\mathbb E(x,z)$ through a hypergeometric function. 

We note that, for each $z\in\mathbb C$, $ {\mathbb E} (\cdot,z)$ is an eigenfunction (belonging to the eigenvalue~$z$) of the annihilation operator $\sigma\partial^-+\eta\partial^+(\partial^-)^2$, which is the adjoint of the operator $\partial^+$. Here $\partial ^+$ and $\partial ^-$ are the raising and lowering operators for the Sheffer sequence $(s_n)_{n=0}^\infty$: \begin{equation}\label{vgdtdsy6}\partial^+s_n=s_{n+1},\quad  \partial^-s_n=ns_{n-1}\quad n\in\mathbb N_0.
\end{equation}

For each $\zeta\in\mathbb C$, we define a complex-valued Poisson measure on $\mathbb N_0$ with parameter~$\zeta$ by 
\begin{equation}\label{vcrew43}
\pi_\zeta (d\xi) = e^{-\zeta} \sum_{n=0}^\infty  \frac1{n!}\, \zeta^n\, \delta_n (d\xi).
\end{equation}
 We prove that the nonlinear coherent states can be written in the form 
$ {\mathbb E} (x,z)= \int_{ {\mathbb N}_0} {\mathcal E} (x, \beta\xi) \, \pi_{\frac z\beta} (d\xi)$,
where 
$$ {\mathcal E} (x, \beta\xi) = \sum_{n=0}^\infty \frac{\beta^n(\xi)_n}{n! \, (\sigma \mid - \eta)_n} \, s_n (x), $$	 
and we derive explicit formulas for ${\mathcal E} (x, \beta\xi)$.	 

Furthermore, in the cases of the gamma distribution and the negative binomial distribution, we prove that, for each $f\in L^2(\mu_{\alpha,\beta,\sigma})$, 
$$(\mathbb Sf)(z)=\int_{\mathbb N_0}\int f(x)\,\mu_{\alpha, \beta,\eta \xi + \sigma}(dx)\,
\pi_{\frac z\beta } (d\xi),\quad z\in\mathbb C.$$
In particular, for $z>0$, $(\mathbb Sf)(z)$ is the expectation of $f$ with respect to the random measure $\mu_{\alpha, \beta,\eta \xi + \sigma} $, where the random variable $\xi$ has Poisson distribution $\pi_{\frac z\beta}$. \
Similarly, in the case of the Meixner distribution, we show that 
$$(\mathbb Sf)(z)=\int_{\mathbb N_0}\int f(x+\beta\xi)\,\mu_{\alpha, \beta,\eta \xi + \sigma}(dx)\,
\pi_{\frac z\beta } (d\xi),\quad z\in\mathbb C.$$
However, this formula holds only for functions $f$ from $\mathcal E_{\text{min}}^1(\mathbb C)$, the space of entire functions of order at most 1 and minimal type \cite{Grabiner}.  (The set $\mathcal E_{\text{min}}^1(\mathbb C)$ is dense in $L^2(\mu_{\alpha,\beta.\sigma})$.) 
Note that, for $r>0$, $(\mathbb Sf)(\beta r)$  is the expectation of the function $f(x+\beta \xi)$ with respect to the probability measure $\mu_{\alpha, \beta,\eta \xi + \sigma}(dx)\pi_r(d\xi) $.

Similarly to the Gaussian and Poisson cases, under the generalized Segal--Bargmann transform $\mathbb S$, the operator of multiplication by the variable $x$ in $L^2(\mu_{\alpha,\beta,\sigma})$ goes over to an operator in $\mathbb F_{\eta,\sigma}(\mathbb C)$ that admits a representation through the operators $Z$ and~$D$.  
	
Let us now briefly describe our strategy to prove these results. Let $\mathcal P(\mathbb C)$ denote the vector space of polynomials over $\mathbb C$. Consider the polynomials $s_n$ as elements of~$\mathcal P(\mathbb C)$ (with real coefficients), and consider $\partial^+$ and $\partial^-$ as linear operators in  $\mathcal P(\mathbb C)$ defined by~\eqref{vgdtdsy6}.  Define linear operators $U$ and $V$ in $\mathcal P(\mathbb C)$ by
\begin{equation}\label{cxdts5yw}
U=\partial^++\beta\partial^+\partial^-+\frac\sigma\alpha,\quad V=\alpha\partial^-+1.
\end{equation}
Let also $Z$ denote the operator of multiplication by variable $z$ in $\mathcal P(\mathbb C)$. In view of~\eqref{545-980+polk}, we get,  in the case $\alpha\ge\beta>0$ (hence $l=\frac\sigma\alpha$): $Z=UV$. Similarly, in the case $\Re(\alpha)\ge0$, $\Im(\alpha)>0$, $\beta=\overline \alpha$  (hence $l=0$), we have  $Z+\frac\sigma\alpha=UV$. Since $[\partial^-,\partial^+]=1$, the operators $U$ and $V$ satisfy the commutation relation 
		\begin{equation}\label{vgdrydydy}[V,U] = \beta V + (\alpha - \beta).
		\end{equation}		
Hence, they generate a generalized Weyl algebra, see e.g.\ \cite[Chapter 8]{MansourSchork} and the references therein.

Consider the linear bijective operator $\mathcal S$ in $\mathcal P(\mathbb C)$ that satisfies $(\mathcal Ss_n)(z)=(z\mid\beta)_n$ ($n\in\mathbb N_0$), see \eqref{cyrtdstevf}. Define operators $\mathcal U=\mathcal SU\mathcal S^{-1}$ and $\mathcal V=\mathcal SV\mathcal S^{-1}$. An easy calculation shows that 
\begin{equation}\label{vtdtstswt}
\mathcal U = Z + \frac{\sigma}{\alpha},\quad \mathcal V = \alpha D_\beta +  1,
\end{equation}
compare with \eqref{mn67-0tr4}.
Obviously, $\mathcal U$ and $\mathcal V$ also satisfy the commutation relation 
$[\mathcal V,\mathcal U] = \beta \mathcal V + (\alpha - \beta)$. Hence, they also generate a generalized Weyl algebra. Compare it with Feinsilver's finite difference algebra \cite{Feinsilver}.

Let us remark that orthogonal Sheffer  sequences with $\eta > 0$ already appeared in studies related to the square of white noise algebra, see e.g. \cite{AccardiFranzSkeide} and the references therein. It was shown in \cite{1-Accardi} that the square of white noise algebra contains a subalgebra generated by elements fulfilling the relations of Feinsilver's finite difference algebra, see also \cite{Boukas1} and \cite{Boukas2}. For further studies of Lie algebras related to orthogonal Sheffer sequences, see \cite{Asai1}, \cite[Appendix~A]{Asai2}, and \cite{Feinsilver2}.

	Similarly to Katriel's theorem about the normal ordering in the Weyl algebra \cite{10-Katriel}, we discuss the normal (Wick) ordering for the operator $(UV)^n$ in terms of $U^k$ and $ V^k$, compare with \cite[Section~8.2]{MansourSchork} and the references therein. This  allows us to derive explicit formulas for $s_n (z)$ and a representation of monomials $z^n$ through the polynomials $s_k (z)$. In these formulas, we use Stirling numbers and Lah numbers. As a corollary, we find useful formulas for the moments of the orthogonality measure $\mu_{\alpha,\beta,\sigma}$. These results are presented in the Appendix, and the reader may find them of independent interest. 

We explicitly construct an open unbounded domain ${\mathcal D}_{\alpha, \beta, \sigma}$ in $\mathbb C$ that contains~$0$. We define a reproducing kernel Hilbert space ${\mathcal F}_{\alpha, \beta, \sigma}$ of analytic functions on ${\mathcal D}_{\alpha, \beta, \sigma}$ that have representation $\varphi(z) = \sum_{n=0}^\infty f_n (z \mid \beta)_n$ with coefficients $f_n \in \mathbb C$ satisfying 
\eqref{vctresay5}.  We extend $\mathcal S$ to a unitary operator $\mathcal S : L^2 (\mu_{\alpha, \beta, \sigma}) \to {\mathcal F}_{\alpha, \beta, \sigma} $ that satisfies $(\mathcal S s_n)(z) = (z \mid \beta)_n$. Thus, under the unitary operator $\mathcal S$, the operator of multiplication by the variable $x$ in $L^2(\mu_{\alpha,\beta,\sigma})$ goes over to the operator $\mathcal U\mathcal V$ in ${\mathcal F}_{\alpha, \beta, \sigma}$ for $\alpha\ge\beta>0$ and to the operator $\mathcal U\mathcal V-\frac\sigma\alpha$ for $\Re(\alpha)\ge0$, $\Im(\alpha)>0$, $\beta=\overline\alpha$. 
We study the unitary operator $\mathcal S$ by using the results obtained through the normal ordering in the generalized Weyl algebras.

Next, we construct a unitary operator 
		$ \mathbb T : {\mathcal F}_{\alpha, \beta, \sigma}  \to {\mathbb F}_{ \sigma,\eta} (\mathbb C) $
	that satisfies
$$ (\mathbb T (\mathbf{\cdot} \mid \beta )_n)(z) = z^n,\quad n\in\mathbb N_0 .$$
 We prove that this operator has a representation 
	 \begin{equation}		\label{87.iur7.k}
		(\mathbb T f)(z) = \int_{\mathbb N_0} f(\beta\xi) \, \pi_{ \frac z\beta } (d\xi),\quad f\in {\mathcal F}_{\alpha, \beta, \sigma}.\ z\in\mathbb C .
			\end{equation}
Finally, we use that $\mathbb S=\mathbb T\mathcal S$.

As a consequence of our considerations, we also derive explicit formulas for the action of the operators $U$ and $V$, defined by \eqref{cxdts5yw}. Compare with \cite[Section~4]{Lytvynov}.

The paper is organized as follows. In Section~\ref{8765rtg00vbo}, we define and discuss the Fock space ${\mathbb F}_{ \eta, \sigma} (\mathbb C)$ and the topological space of entire functions $\mathcal E_{\mathrm{min}}^1(\mathbb C)$. In Section~\ref{crtds6ewu65e3}, we present our main results. In Section~\ref{vftds5a453q}, we present the proofs of the main results. Finally, in the Appendix~A,  
we discuss   the normal ordering in the generalized  Weyl algebra generated by operators $U$, $V$ satisfying the commutation relation $[V,U]=aV+b$ with $a,b\in\mathbb C$. We apply the obtained result to an orthogonal Sheffer sequences $(s_n)_{n=0}^\infty$ and find useful formulas for the moments of its orthogonality measure. 
		
We expect that the key ideas of this paper can be extended to an infinite-dimensional setting, compare with  \cite{Lytvynov}. This will be a topic of our future research.

\section{The spaces ${\mathbb F}_{ \eta, \sigma} (\mathbb C)$ and $\mathcal E_{\mathrm{min}}^1(\mathbb C)$} 	\label{8765rtg00vbo}

For $\eta>0$ and $\sigma\ge0$, we denote by $\, {\mathbb F}_{ \eta, \sigma} (\mathbb C)$ the vector space of all entire functions 
$\varphi : \mathbb C \to \mathbb C$, $\varphi (z) = \sum_{n=0}^{\infty} f_n z^n$ 
with coefficients $ f_n \in \mathbb C$ ($n\in\mathbb N_0$) satisfying \eqref{vctresay5}.   Consider ${\mathbb F}_{ \eta, \sigma} (\mathbb C)$ as a Hilbert space equipped with the inner product 
$ (\varphi, \psi )_{{\mathbb F}_{ \eta, \sigma}(\mathbb C) } =\sum_{n=0}^{\infty} f_n \, \overline{g_n} \, n! \, (\sigma \mid -\eta)_n $ 
for $ \varphi(z) = \sum_{n=0}^{\infty} \, f_n z^n, \ \psi(z) = \sum_{n=0}^{\infty} \, g_n z^n \in {\mathbb F}_{ \eta, \sigma} (\mathbb C) $. This is a reproducing kernel Hilbert space with reproducing kernel 
${\mathbb K} (z,w) = \sum_{n=0}^\infty \frac{(\bar zw)^n}{n!\,(\sigma\mid-\eta)_n}$,
i.e., for each $\varphi\in {\mathbb F}_{ \eta, \sigma}(\mathbb C) $, we have 
$(\varphi,{\mathbb K} (z,\cdot))_{{\mathbb F}_{ \eta, \sigma}(\mathbb C) }=\varphi(z)$. 

Consider the following gamma distribution on $\R_+$:
$$\mu_{\eta,\eta,\, \eta \sigma} (dr) = \frac{1}{\Gamma (\frac{\sigma}{\eta})}  \bigg(\frac{1}{\eta}\bigg)^{\frac{\sigma}{\eta}}r^{-1+ \frac{\sigma}{\eta}} e^{- \frac{r}{\eta}} \, dr.$$
Let $\lambda_{\eta,\sigma}$ be the random Gaussian measure $\nu_r$ (see formula~\eqref{zsawsa}) where the 
 random variable $r$ is  distributed according to~$\mu_{\eta,\eta, \eta \sigma}$, i.e., 
\begin{equation}\label{vctsaq4t}
\lambda_{\eta, \sigma} (dz)= \int_{\R_+}  \nu_r (dz)\,\mu_{\eta,\eta,\, \eta \sigma} (dr)  = \Lambda_{\eta, \sigma}(z) \, A(dz),\end{equation}
 where
 \begin{equation}
\Lambda_{\eta, \sigma}(z)=\frac{1}{\pi\,\Gamma (\frac{\sigma}{\eta})} \bigg(\frac{1}{\eta}\bigg)^{\frac{\sigma}{\eta}} \int_{\R_+} \exp\left(-\frac{|z|^2}r-\frac r\eta\right)r^{-2+ \frac{\sigma}{\eta}} \, dr.\label{xdsresar5}
 \end{equation}
 
 In the following proposition,  we will use the modified Bessel function
 \[
 K_\theta(x)=\frac{\pi}{2\sin(\theta\pi)}\big(I_{-\theta}(x)-I_{\theta}(x)\big),\]
 where
  \[
  I_\theta(x)=\bigg(\frac x2\bigg)^\theta\sum_{n=0}^\infty\frac{(x/2)^{2n}}{n!\,\Gamma(\theta+n+1)}.
  \]
  In these formulas, the parameter $\theta$ is assumed to be not an integer. When $\theta$ is an integer, the limit is used to define $K_\theta(x)$.

\begin{proposition}\label{ytd6d6de6}		 Let $\eta>0$  and  $\sigma>0$.   Then ${\mathbb F}_{ \eta, \sigma} (\mathbb C)$ is the closed subspace of $ L^2 (\mathbb C, \lambda_{\eta, \sigma})$ constructed as the closure of $\mathcal P( \mathbb C)$. Furthermore,
\begin{equation}\label{bhftrdrs465u}
\Lambda_{\eta, \sigma}(z)=\frac{2\eta^{-\frac12(1+\frac\sigma\eta)}}{\pi\Gamma(\frac\sigma\eta)}\,|z|^{\frac\sigma\eta-1}\,K_{1-\frac\sigma\eta}\big(2\eta^{-\frac12}|z|\big).
\end{equation}
\end{proposition}

\begin{proof} Recall that, for $m,n\in\mathbb N_0$, we have $\int_{\mathbb C}z^m\,\overline{z^n}\,\nu_r(dz)=\delta_{n,m}\,r^n\,n!$\,. By formula \eqref{cfgdyjk} in the Appendix, we get $\int_{\R_+}r^n \, \mu_{\alpha, \alpha, \sigma} (dr) = \left( \frac{\sigma}{\alpha} \mid -\alpha \right)_n $. Hence, by~\eqref{vctsaq4t}, 
\begin{align*}
\int_{\mathbb C} z^m \, \overline{z^n} \, \lambda_{\eta, \sigma} (dz) 
&= \int_{\R_+} \int_{\mathbb C} z^m \, \overline{z^n} \,  \nu_r (dz) \, \mu_{\eta,\eta,\, \eta\sigma} (dr) \\
&= \delta_{m,n}\,n! \int_{\R_+} r^n\, \mu_{\eta,\eta,\, \eta \sigma} (dr) = \delta_{m,n}\,n!  \left(\sigma \mid -\eta \right)_n . 
\end{align*} 
Formula \eqref{bhftrdrs465u} for the density $\Lambda_{\eta, \sigma}(z)$ of the measure $\lambda_{\eta, \sigma}$ was proved by Asai \cite[Theorem~3.1]{Asai1}.
\end{proof}

\begin{remark} In fact, $\lambda_{\eta, \sigma}$ is the unique probability measure on $\mathbb C$ which satisfies 
$$\int_{\mathbb C} z^m \, \overline{z^n} \, \lambda_{\eta, \sigma} (dz)=\delta_{m,n}\,n!  \left(\sigma \mid -\eta \right)_n,$$ see \cite[Theorem~3.1]{Asai1}. 

\end{remark}

Following \cite{Alpay}, let us recall some basic facts about the Mellin transform and the Mellin convolution.
Let $f : \R_+ \to \mathbb R$ be such that, for some interval $(a, b) \subset \mathbb R$, the function $f(r)  r^{c-1}$ is integrable on $\R_+$ for all $c \in (a,b)$. Then the Mellin transform of $f$ is defined by $\mathcal M (f)(c) = \int_{\R_+} r^{c-1}  f(r) \, dr$ for $c \in (a,b)$.
Obviously, for $\eta > 0$ and $f(r) = e^{-r/\eta}$, we have $\mathcal M (f)(c)= \eta^c\,  \Gamma(c) $ for $c>0$.
The Mellin convolution of functions $f$ and $g$ is the function $f \ast g$ that satisfies 
$ \mathcal M (f \ast g) (c) = \mathcal M (f)(c)\,\mathcal M (g)(c)$.
Explicitly, the function $f \ast g$ is given by 
\begin{equation}
	(f \ast g) (r) = \int_{\R_+} f \left( \frac{r}{t} \right)  g(t) \, \frac{1}{t} \, dt = \int_{\R_+} f(t)   g \left( \frac{r}{t} \right) \frac{1}{t} \, dt, \quad r>0. \label{875tgh} 
						       \end{equation}
\begin{lemma} \label{87654fh}
Assume that $\eta=\sigma $. Then the function $\Lambda_{\sigma, \sigma}$ in Proposition~\ref{ytd6d6de6} has the form
$ \Lambda_{\sigma, \sigma} (z) =(\pi \sigma)^{-1}  \psi(\lvert z \rvert^2)$,
where $ \psi (r) = (f_1 \ast f_2)(r)$
with $f_1 (r) = e^{-r}$ and $f_2 (r) = e^{-r/\sigma}$. 
\end{lemma}

\begin{proof}
Immediate by formulas \eqref{xdsresar5} and \eqref{875tgh}.
\end{proof}

\begin{remark}
In the special case $\eta=\sigma = 1$, the statement of Lemma~\ref{87654fh} was proved in 
 in \cite{Alpay-8, Alpay}. 
  By \cite[p.~5]{Alpay}, $\psi(r)= \int_{\mathbb R} \exp(- \sqrt r \, 2 \cosh (x)) \, dx$ is a modified Bessel function of the second kind. 
\end{remark}

Let $\varphi: \mathbb C \rightarrow \mathbb C$ be an entire function. One says that $\varphi$ is of order at most 1  and minimal type (when the order is equal to 1) if $\varphi$ satisfies  
$$ \sup_{z \in \mathbb C} |\varphi(z)| \exp(-t \, |z|) < \infty\quad \forall t>0.$$
One denotes by ${\mathcal E}_{\mathrm{min}}^1 (\mathbb C)$ the vector space of all such functions.

For each  $t > 0$, $\lVert \varphi \rVert_t = \sup_{z \in \mathbb C} |\varphi(z)| \exp(-t \, |z|)$ is a norm  on $ {\mathcal E}_{\mathrm{min}}^1 (\mathbb C)$, and denote by $B_t$ the completion of $ {\mathcal E}_{\mathrm{min}}^1 (\mathbb C)$ in this norm. For any $0<t_1<t_2$, the Banach space $B_{t_1}$ is continuously embedded into $B_{t_2}$. Note that, as a set, $ {\mathcal E}_{\mathrm{min}}^1 (\mathbb C) = \bigcap_{t>0} B_{t}$. One defines the projective topology on ${\mathcal E}_{\mathrm{min}}^1 (\mathbb C)$ induced by the $B_t$ spaces, i.e., one chooses the coarsest locally convex topology on ${\mathcal E}_{\mathrm{min}}^1 (\mathbb C)$ for which the embedding of ${\mathcal E}_{\mathrm{min}}^1 (\mathbb C)$ into  $B_{t}$ is continuous for each $t > 0$. Equipped with this topology, $ {\mathcal E}_{\mathrm{min}}^1 (\mathbb C)$ is a  Fr\'echet space. The following theorem is proved by Grabiner \cite{Grabiner}, see also \cite{Sheffer}.

\begin{theorem}[\cite{Grabiner}]\label{bvgytdddsesea}
Let $(s_n)_{n=0}^\infty$ be a Sheffer sequence with generating function \eqref{vcxresayw}. Assume that the formal power series $A(t)$ and $B(t)$ in~\eqref{vcxresayw} determine analytic functions in a neighborhood of zero. Then the following statements hold.

(i)  An entire function  $\varphi:\mathbb C\to\mathbb C$ belongs to $ {\mathcal E}_{\mathrm{min}}^1 (\mathbb C)$ if and only if it can be represented in the form 
\begin{equation} \label{uniqueRep}
\varphi(z) = \sum_{n=0}^\infty f_n \, s_n (z) ,
\end{equation}
where $\sum_{n=0}^\infty |f_n|^2 \, (n!)^{2} \, 2^{nk} < \infty$ for all $k \in \mathbb N$. The representation of the function $\varphi$ as in \eqref{uniqueRep} is unique, and the series on the right-hand side of formula \eqref{uniqueRep} converges in $ {\mathcal E}_{\mathrm{min}}^1(\mathbb C)$.

(ii) For each $k \in \mathbb N$, denote by $\mathcal H_{k}$ the completion of $ {\mathcal E}_{\mathrm{min}}^1(\mathbb C)$ in the Hilbertian norm ${\vertiii{\varphi}}_{k} = \left( \sum_{n=0}^\infty |f_n|^2 \, (n!)^2 \, 2^{nk} \right)^{1/2}$,  where $f_n$ ($n \in {\mathbb N}_0$) are the coefficient from \eqref{uniqueRep}. Then, ${\mathcal E}_{\mathrm{min}}^1(\mathbb C)$ is the projective limit of the $\mathcal H_k$ spaces. 

\end{theorem}

\begin{corollary} (i) For each $\eta\ge0$ and $\sigma>0$, the Fr\'echet space ${\mathcal E}_{\mathrm{min}}^1(\mathbb C)$ is continuously embedded into ${\mathbb F}_{ \eta, \sigma} (\mathbb C)$.

(ii) Let $(s_n)_{n=0}^\infty$ be an orthogonal Sheffer sequence and let $\mu_{\alpha,\beta,\sigma}$ be its orthogonality measure. Then the Fr\'echet space ${\mathcal E}_{\mathrm{min}}^1(\mathbb C)$ is continuously embedded into $L^2(\mu_{\alpha,\beta,\sigma})$. Furthermore, ${\mathcal E}_{\mathrm{min}}^1(\mathbb C)$ is a dense subset of  $L^2(\mu_{\alpha,\beta,\sigma})$.

\end{corollary}

\begin{proof}(i) The sequence of monomials $(z^n)_{n=0}^\infty$ is a Sheffer sequence for which $A(t)=0$ and $B(t)=t$, hence it satisfies the conditions of Theorem~\ref{bvgytdddsesea}. Therefore, the statement follows from the definition of ${\mathbb F}_{ \eta, \sigma} (\mathbb C)$ and Theorem~\ref{bvgytdddsesea}. 

(ii) It follows from \cite{13-Meixner} that each orthogonal Sheffer sequence satisfies the conditions of Theorem~\ref{bvgytdddsesea}. Next, it follows from the recurrence formula \eqref{545-980+polk} that $\|s_n\|_{L^2(\mu_{\alpha,\beta,\sigma})}^2=n!\,(\sigma\mid-\eta)_n$. Hence, $\varphi\in L^2(\mu_{\alpha,\beta,\sigma})$ if and only if $\varphi(x)=\sum_{n=0}^\infty f_n s_n(x)$ with $f_n$ sa\-tis\-fying \eqref{vctresay5}, and the series $\sum_{n=0}^\infty f_n s_n(x)$ converges in $L^2(\mu_{\alpha,\beta,\sigma})$. Since $(\sigma\mid-\eta)_n\le n!\, (\min\{\eta,\sigma\})^n$, the statement   follows from Theorem~\ref{bvgytdddsesea}.
\end{proof}

\section{Main results}\label{crtds6ewu65e3}
Let $\sigma>0$. We assume that either $\alpha\ge\beta>0$ and $l=\sigma/\alpha$ or  $\Re(\alpha)\ge 0$, $\Im(\alpha)>0$, $\beta=\overline \alpha$, and $l=0$. Let $(s_n)_{n=0}^\infty$ be the Sheffer sequence satisfying the recurrence formula~\eqref{545-980+polk}, and let $\mu_{\alpha,\beta,\sigma}$ be its orthogonality measure. We denote by $X_{\alpha,\beta}$ the support of $\mu_{\alpha,\beta,\sigma}$, i.e., $X_{\alpha,\beta}=\R_+$ if $\alpha=\beta>0$, $X_{\alpha,\beta}=(\alpha-\beta)\mathbb N_0$ if $\alpha>\beta>0$, and $X_{\alpha,\beta}=\R$ if $\Re(\alpha)\ge0$, $\Im(\alpha)>0$, $\beta=\overline \alpha$. 

We define a generalized Segal--Bargmann transform $ \mathbb S:L^2(X_{\alpha,\beta},\mu_{\alpha,\beta,\sigma})\to\mathbb F_{\eta,\sigma}(\mathbb C)$ as a unitary operator satisfying $(\mathbb Ss_n)(z)=z^n$ for $n\in\mathbb N_0$. 

\begin{theorem}\label{vcfxsra5u6e756}
 The generalized Segal--Bargmann transform $ \mathbb S$ has a representation through 
the nonlinear coherent states 		\begin{equation} 	
\mathbb  E (x, z) = \sum_{n=0}^\infty \frac{z^n}{n!\,(\sigma\mid-\eta)_n} \, s_n (x), \quad x\in X_{\alpha,\beta},\ z\in\mathbb C,\label{ydtyrdsj6ej}
\end{equation}
i.e., $\mathbb E(\cdot,z)\in L^2(X_{\alpha,\beta},\mu_{\alpha,\beta,\sigma})$ for each $z\in\mathbb C$ and 
\begin{equation}\label{ctstesa5a5}
(\mathbb S f)(z) =\int_{X_{\alpha,\beta}} f(x)\,\mathbb  E(x, z) \,\mu_{\alpha,\beta,\sigma}(dx),\quad z\in\mathbb C.\end{equation}
Furthermore, if $\alpha=\beta>0$,
\begin{equation}\label{cfsersara5s}
\mathbb  E(x, z) = \int_{\mathbb N_0}\big[(\sigma/\eta)^{(\xi)}\big]^{-1}
\bigg(\frac x\alpha\bigg)^\xi \, \pi_{\frac z \alpha } (d\xi),\quad x\in\R_+,\ z\in\mathbb C,
\end{equation}
if $\alpha>\beta>0$, 
\begin{equation}\label{ctsxtersrtesaeas}
\mathbb  E((\alpha-\beta)n, z)=\int_{\mathbb N_0}{\left( 1 - \frac{\beta}{\alpha} \right)}^{\xi} \, \frac{(\xi\eta + \sigma \mid - \eta)_n}{(\sigma \mid - \eta)_n}\,\pi_{\frac z\beta} (d\xi),\quad n\in\mathbb N_0,\ z\in\mathbb C,
\end{equation}
and if $\Re(\alpha)\ge0$, $\Im(\alpha)>0$, $\beta=\overline \alpha$,
\begin{align}
 \mathbb  E(x, z) &= \int_{\mathbb N_0}\left( 2\cos\big(\frac\pi2 - \operatorname{Arg}(\alpha) \big) \right)^{\xi} \left((\sigma/\eta)^{(\xi)}\right)^{-1}  \exp\left( i(\frac\pi2-\operatorname{Arg}(\alpha))\xi \right)\notag \\
&\quad\times  \left(-\frac{ix}{2 \Im(\alpha)} -\frac{i\sigma\alpha}{2\eta \Im(\alpha)}\right)^{(\xi)}\,\pi_{\frac z\beta} (d\xi),\quad x\in\R,\ z\in\mathbb C,\label{xawawawefljkreiojfu}
\end{align}
where $ \operatorname{Arg}(\alpha)\in[0,\pi/2)$.
\end{theorem}

Let $\partial^+$ and $\partial^-$ denote the raising and lowering operators for the Sheffer sequence $(s_n)_{n=0}^\infty$, see \eqref{vgdtdsy6}.
 Denote $A^-=\sigma\partial^-+\eta\partial ^+(\partial^-)^2$.

\begin{corollary}\label{vcsara5a54w45ss} For any $p,q\in\mathcal P(\mathbb C)$, $(\partial^+p,q)_{L^2(\mu_{\alpha,\beta,\sigma})}=(p,A^-q)_{L^2(\mu_{\alpha,\beta,\sigma})}$ and the operator $A^-$ with domain $\mathcal P(\mathbb C)$ is closable in $L^2(\mu_{\alpha,\beta,\sigma})$.  Keep the notation $A^-$ for the closure of~$A^-$. Then, for each $z\in\mathbb C$,  $\mathbb E(\cdot,z)$ is an eigenvector of $A^-$ belonging to the eigenvalue $z$.

\end{corollary}

For $\alpha\ge\beta>0$ and $z\in\mathbb C$, we define a complex-valued measure $\rho_{\alpha,\beta,\sigma,z}$ on $X_{\alpha,\beta}$ by 
\begin{equation}\label{vr5y365re5e54}		
\rho_{\alpha,\beta,\sigma,z}(dx)=	\int_{\mathbb N_0}\,	\mu_{\alpha, \beta,\eta \xi + \sigma} (dx)\,\pi_{\frac z\beta } (d\xi).
\end{equation}
In particular, if  $z>0$, $\rho_{\alpha,\beta,\sigma,z}$ is the random measure $\mu_{\alpha, \beta,\eta \xi + \sigma} $, where the random variable $\xi$ has Poisson distribution $\pi_{\frac z\beta}$\,.

\begin{theorem}\label{cdszrea5r4w35}
Let $\alpha\ge\beta>0$. For each $f \in L^2 (X_{\alpha,\beta},\mu_{\alpha, \beta, \sigma})$,
\begin{equation}\label{vcdrt5e}
	(\mathbb Sf)(z)=\int_{X_{\alpha,\beta}} f(x)\,	\rho_{\alpha,\beta,\sigma,z}(dx),\quad z\in\mathbb C.
	\end{equation}
\end{theorem}

In the case where $\alpha$ and $\beta$ have non-zero imaginary part, a counterpart of Theorem~\ref{cdszrea5r4w35} has the following form. 

\begin{theorem}\label{xtxstesar5ewa54aqy}
Let $\Re(\alpha)\ge0$, $\Im(\alpha)>0$, $\beta=\overline \alpha$. The operator $\mathbb S$, considered as a linear operator in $\mathcal P(\mathbb C)$, admits an extension to a continuous linear operator $\mathbb S$ in  $\mathcal E^1_{\mathrm{min}}(\mathbb C)$, and for each $f \in \mathcal E^1_{\mathrm{min}}(\mathbb C)$	and $z\in\mathbb C$,
$$(\mathbb Sf)(z)=\int_{\mathbb N_0}\int_{\R}f(x+\beta\xi)\,\mu_{\alpha,\beta,\eta\xi+\sigma}(dx)\,\pi_{\frac z\beta}(d\xi). $$
In particular, for each $r>0$,
$$(\mathbb Sf)(\beta r)=\int_{\mathbb N_0}\int_{\R}f(x+\beta\xi)\,\mu_{\alpha,\beta,\eta\xi+\sigma}(dx)\,\pi_{r}(d\xi). $$
\end{theorem}

	\begin{remark}Using the approach to the generalized Segal--Bargmann transform developed in this paper, one can easily show that, in the case of the monic Charlier polynomials $(c_n)_{n=0}^\infty$ that are orthogonal with respect to the Poisson distribution~$\pi_\sigma$ ($\sigma>0$, $\alpha=1$, $\beta=0$, $l=\sigma$), the  corresponding Segal--Bargmann transform $\mathbb S:L^2(\mathbb N_0,\pi_\sigma)\to \mathbb F_\sigma(\mathbb C)$, satisfying $\mathbb Sc_n=z^n$ ($n\in\mathbb N_0$),  admits the following representation:  
\begin{equation*}
(\mathbb Sf)(z)=\int_{\mathbb N_0}f(x)\,\pi_{\sigma+z}(dx),\quad f\in L^2(\mathbb N_0,\pi_\sigma),\ z\in\mathbb C,
\end{equation*}
compare with formula  \eqref{fdzraer}, which holds in the Gaussian case. Note that, for $z\in(-\sigma,+\infty)$, $\pi_{\sigma+z}$ is the (usual) Poisson distribution with parameter $\sigma+z$.	
	\end{remark}

We define linear operators $\mathbb U=Z+\beta ZD+\frac\sigma\alpha$ and $\mathbb V=\alpha D+1$, acting in $\mathcal P(\mathbb C)$ and satisfying $[\mathbb V,\mathbb U]= \beta \mathbb  V + (\alpha - \beta)$. We also define  the operator 
$$ \rho=\mathbb U\mathbb V=Z+\lambda ZD+\frac\sigma\alpha+\sigma D+\eta ZD^2.$$

\begin{proposition}\label{cfxctdxtsdrytdsy}
The operator $\rho$ is essentially self-adjoint in $\mathbb F_{\eta,\sigma}(\mathbb C)$ and we keep the notation $\rho$ for its closure. If $\alpha\ge\beta>0$, then $\mathbb S\rho\mathbb S^{-1}$ is the operator of multiplication by the variable $x$ in $L^2(X_{\alpha,\beta},\mu_{\alpha,\beta,\sigma})$.  If $\Re(\alpha)\ge0$, $\Im(\alpha)>0$ and $\beta=\overline \alpha$, then $\mathbb S(\rho-\frac\sigma\alpha)\mathbb S^{-1}$ is the operator of multiplication by the variable $x$ in $L^2(\R,\mu_{\alpha,\beta,\sigma})$.
\end{proposition}

The proof of the above statements will be based on Lemmas~\ref{yufd6e64w5}--\ref{ftsxrearwa} below.  
		 
		Similarly to \eqref{vcrew43}, we will now define a complex-valued measure $\mu_{\alpha,\beta,\zeta}$ for a complex parameter $\zeta$.  First, we define a domain $\mathfrak D_{\alpha,\beta}$ in $\mathbb C$ as follows. If $\alpha>\beta>0$, we define $\mathfrak D_{\alpha,\beta}=\mathbb C$, if either $\alpha=\beta>0$ or $\Re(\alpha)=0$, $\Im(\alpha)>0$,  $\beta=\overline\alpha$, we define $\mathfrak D_{\alpha,\beta}=\{z\in\mathbb C\mid \Re(z)>0\}$, and if $\Re(\alpha)>0$, $\Im(\alpha)>0$,  $\beta=\overline\alpha$, we define
\begin{equation}\label{cts5s5w5}
		\mathfrak D_{\alpha, \beta} = \big\{ z\in \mathbb C \mid \Re(z) > 0,\  \lvert \Im(z) \rvert < \Re(z)\Im(\alpha)/\Re(\alpha) \big\}.
	\end{equation}	
	Now, if $\alpha\ge\beta>0$ and $\zeta\in \mathfrak D_{\alpha,\beta}$, we define the complex-valued measure $\mu_{\alpha,\beta,\zeta}$	on $X_{\alpha,\beta}$ by replacing the positive parameter $\sigma$ in formulas \eqref{buyfr7iei9} and  \eqref{buyfbufyt7r7iei9} with $\zeta$. Next, if $\Re(\alpha)\ge0$, $\Im(\alpha)>0$, $\beta=\overline\alpha$, we use the formula $\overline{\Gamma(z)}=\Gamma(\overline{z})$ for $z\in\mathbb C$, $\Re(z)>0$,  to write formula \eqref{cftstew5}
 in the form 
 \begin{align}
	 {\mu}_{\alpha, \beta, \sigma}(dx) &={C}_{\alpha, \beta, \sigma} \, \exp\left( (\pi/2 - \operatorname{Arg}(\alpha))x/\Im(\alpha) \right) \notag\\
	 &\quad\times\Gamma \left( \frac{ix}{2\Im(\alpha)} + \frac{i\sigma\beta}{2\eta\Im(\alpha)} \right)
	 \Gamma \left( -\frac{ix}{2\Im(\alpha)} - \frac{i\sigma\alpha}{2\eta\Im(\alpha)} \right)  dx.	
\label{cdzserar4ftstew5}\end{align} 
Now, for $\zeta\in \mathfrak D_{\alpha,\beta}$, we define the complex-valued measure $\mu_{\alpha,\beta,\zeta}$ on $\R$ by replacing the positive parameter $\sigma$ in formulas \eqref{sewa4tq}, \eqref{cdzserar4ftstew5} with $\zeta$.
		
		Furthermore, we define an open domain $\mathcal D_{\alpha,\beta,\sigma}$ in $\mathbb C$ as follows. If $|\alpha|=|\beta|$ (i.e., either $\alpha=\beta>0$ or $\Re(\alpha)\ge0$, $\Im(\alpha)>0$, $\beta=\overline\alpha$),
\begin{equation}\label{tedrer4we4}
\mathcal D_{\alpha,\beta,\sigma} = \big\{ z \in \mathbb C\mid \Re (\alpha z) > -\sigma/2\big\},\end{equation}
and if $\alpha>\beta>0$, $\mathcal D_{\alpha,\beta,\sigma}=\mathbb C$. We will use below the following obvious observation: for each $z\in \mathcal D_{\alpha,\beta,\sigma}$ and $n\in\mathbb N$, $z+\beta n\in \mathcal D_{\alpha,\beta,\sigma}$. In particular, $\beta\mathbb N_0\subset \mathcal D_{\alpha,\beta,\sigma}$.

\begin{lemma}\label{yufd6e64w5}
Let $(f_n)_{n=0}^\infty$ be a sequence of complex numbers such that \eqref{vctresay5} holds.
Then the series $\sum_{n=0}^\infty f_n (z \mid \beta)_n$ converges uniformly on compact sets in $\, \mathcal D_{\alpha,\beta,\sigma}$, hence it is a holomorphic function on $ \, \mathcal D_{\alpha,\beta,\sigma}$. 
Denote by ${\mathcal F}_{\alpha, \beta, \sigma} $ the vector space 
 of all holomorphic functions on $\mathcal D_{\alpha,\beta,\sigma}$ that have representation
 \begin{equation} \label{65uty0}
\varphi (z) = \sum_{n=0}^{\infty} f_n \, (z \mid \beta)_n , 
\end{equation}
with $(f_n)_{n=0}^\infty$ satisfying \eqref{vctresay5}. Then 
\begin{equation}
f_n = \frac{1}{n!} ( D_\beta^n \varphi ) (0) = \frac{(-1)^n}{n! \, \beta^n}  \sum_{k=0}^n (-1)^k \binom{n}{k} \, \varphi(\beta k) .\label{i78yu}
\end{equation}
In particular, a function $\varphi\in {\mathcal F}_{\alpha, \beta, \sigma}$ has a unique representation \eqref{65uty0}, and $\varphi$ is completely determined by its values on the set $\beta\mathbb N_0$.
\end{lemma}

Let us consider $ {\mathcal F}_{\alpha, \beta, \sigma} $ as a Hilbert space equipped with the inner product  
$(\varphi,\psi)_{ {\mathcal F}_{\alpha, \beta, \sigma} }=\sum_{n=0}^\infty f_n\,\overline{g_n} \, n! \, (\sigma \mid -\eta)_{n}$
for $\varphi(z)=\sum_{n=0}^\infty f_n(z\mid\beta)_n,\,\psi(z)=\sum_{n=0}^\infty g_n(z\mid\beta)_n \in  {\mathcal F}_{\alpha, \beta, \sigma}$\,. Let $\mathcal S:L^2(X_{\alpha,\beta},\mu_{\alpha, \beta, \sigma}) \to  {\mathcal F}_{\alpha, \beta, \sigma} $
  be the unitary operator satisfying
  $$(\mathcal Ss_n)(z)=(z\mid\beta)_n,\quad n\in\mathbb N_0.$$ 
   Define     
  \begin{equation} 	\label{5yrthfg}
			{\mathcal E} (x, z) = \sum_{n=0}^\infty \frac{(z \mid \beta)_n}{n! \, (\sigma \mid - \eta)_n} \, s_n (x), \qquad x \in X_{\alpha,\beta}, \  z \in \mathcal D_{\alpha,\beta,\sigma}.
		\end{equation}
		
		\begin{lemma}\label{vcydyrtds6y}
	For each $z \in\mathcal D_{\alpha,\beta,\sigma}$, we have $\mathcal E (\mathbf{\cdot}, z) \in L^2 (X_{\alpha,\beta},\mu_{\alpha, \beta,\sigma})$ and 
\begin{equation}({\mathcal S} f)(z) = \int_{X_{\alpha,\beta}} f(x)  {\mathcal E} (x, z) \, \mu_{\alpha, \beta,\sigma} (dx),\quad  f\in L^2(X_{\alpha,\beta},\mu_{\alpha,\beta,\sigma}). 	\label{4rggn9}
		  \end{equation}	
Furthermore, if $\alpha=\beta>0$, 
\begin{equation} 	\label{sqr45op}
{\mathcal E}(x, z) = \frac{\Gamma \left( \frac{\sigma}{\alpha^2} \right)}{\Gamma \left( \frac{\alpha z + \sigma}{\alpha^2} \right)} \bigg(\frac x\alpha\bigg)^{\frac{z}{\alpha}},\quad x\in\R_+,\ z\in\mathcal D_{\alpha,\alpha,\sigma}\,,	\end{equation} 
if $\alpha>\beta>0$,
	\begin{equation}	\label{5yrgf7}
		{\mathcal E}((\alpha - \beta)n, z) = \left( 1 - \frac{\beta}{\alpha} \right)^{\frac{z}{\beta}} \, \frac{(\alpha z + \sigma \mid - \eta)_n}{(\sigma \mid - \eta)_n}\quad n\in\mathbb N_0,\  z\in\mathbb C,
	\end{equation}
  and if $\Re(\alpha)\ge0$, $\Im(\alpha)>0$, $\beta=\overline\alpha$,
  \begin{align}
				{\mathcal E}(x, z) 
				&= \left( 2\cos\big(\frac\pi2 - \operatorname{Arg}(\alpha) \big) \right)^{\frac{\alpha z}{\eta }} \frac{\Gamma \left( \frac{\sigma}{\eta } \right)}{\Gamma \left( \frac{\sigma + \alpha z}{\eta } \right)} \exp\left( \frac{i(\frac\pi2-\operatorname{Arg}(\alpha))\alpha z}{\eta } \right)\notag \\
				&\quad\times 
				\frac{\Gamma \left( -\frac{ix}{2 \Im(\alpha)} -\frac{i\sigma\alpha}{2\eta \Im(\alpha)} + \frac{ \alpha z}{\eta }  \right)}{\Gamma \left( -\frac{ix}{2 \Im(\alpha)} -\frac{i\sigma\alpha}{2\eta \Im(\alpha)}   \right)},\quad x\in\R,\ z\in \mathcal D_{\alpha,\beta,\sigma}.\label{fdtrs5ea43qq}
\end{align}
\end{lemma}

The following lemma provides alternative formulas for  the action of the operator~$\mathcal S$.

\begin{lemma}\label{xfsxersarewa4w} (i) Let $\alpha\ge\beta>0$. Then we have, 
  for each $f\in L^2(X_{\alpha,\beta},\mu_{\alpha,\beta,\sigma})$ and  $z \in \mathcal D_{\alpha,\beta,\sigma}$, 
 \begin{equation} \label{7ir7yj}
(\mathcal S f)(z) = \int_{X_{\alpha,\beta}} f(x) \,   \mu_{\alpha, \beta, \alpha z + \sigma}(dx).
\end{equation}

(ii)  Let $\Re(\alpha)\ge0$, $\Im(\alpha)>0$ and $\beta=\overline\alpha$. The operator $\mathcal S$, considered as a linear operator in $\mathcal P(\mathbb C)$, admits an extension to a continuous linear operator $\mathcal S$ in  $\mathcal E^1_{\mathrm{min}}(\mathbb C)$, and for each $f \in \mathcal E^1_{\mathrm{min}}(\mathbb C)$, 
 \begin{equation} \label{ftsertsreyesayj}
(\mathcal S f)(z) = \int_{\R} f(x+z) \,   \mu_{\alpha, \beta, \alpha z + \sigma}(dx),\quad z \in \Psi_{\alpha,\beta,\sigma}\,,
\end{equation}
where
\begin{equation}\label{ctrs5wu65cdrsdr}
\Psi_{\alpha,\beta,\sigma}=\{z\in\mathbb C\mid  \alpha z+\sigma\in \mathfrak D_{\alpha,\beta}\}.
\end{equation} 
\end{lemma} 

Recall the operators $\mathcal U$ and $\mathcal V$, given by \eqref{vtdtstswt} and satisfying $[\mathcal V,\mathcal U]= \beta \mathcal V + (\alpha - \beta)$. 
We define the operator $ \mathcal R=\mathcal U\mathcal V$ acting in $\mathcal P(\mathbb C)$.

\begin{lemma}\label{xdesrea5qw3} The operator $\mathcal R$ is essentially self-adjoint in ${\mathcal F}_{\alpha, \beta, \sigma} $ and we keep the notation $\mathcal R$ for its closure. Then, if $\alpha\ge\beta>0$, $\mathcal S\mathcal R\mathcal S^{-1}$ is the operator of multiplication by the variable $x$ in $L^2(X_{\alpha,\beta},\mu_{\alpha,\beta,\sigma})$, and if $\Re(\alpha)\ge0$, $\Im(\alpha)>0$, $\beta=\overline\alpha$, $\mathcal S(\mathcal R-\frac\sigma\alpha)\mathcal S^{-1}$ is the operator of multiplication by the variable $x$ in $L^2(X_{\alpha,\beta},\mu_{\alpha,\beta,\sigma})$.
\end{lemma}

Next, we define a unitary operator $\mathbb T : {\mathcal F}_{\alpha, \beta, \sigma} \to {\mathbb F}_{\eta, \sigma} (\mathbb C)$ satisfying 
 $$ (\mathbb T (\cdot \mid \beta)_n )(z) = z^n,\quad n\in\mathbb N_0.$$
		
\begin{lemma} \label{ftsxrearwa}
For each $f \in {\mathcal F}_{\alpha, \beta, \sigma}$ and $z \in \mathbb C$, formula~\eqref{87.iur7.k} holds.
\end{lemma}

Recall the operators $U$ and $V$, defined by \eqref{cxdts5yw} and satisfying the commutation relation \eqref{vgdrydydy}. 

\begin{proposition}\label{xszswcxdzwawawa} The operators $U$ and $V$ acting in $\mathcal P(\mathbb C)$ can be (uniquely) extended to continuous linear operators acting in $\mathcal E_{\mathrm{min}}^1(\mathbb C)$. We preserve the notations $U$ and $V$ for these extensions. Let also $Z$ denote the continuous linear operator in $\mathcal E_{\mathrm{min}}^1(\mathbb C)$ of multiplication by variable $z$,
If $\alpha\ge\beta>0$, then  $Z=UV$ and $U=Z(1-\alpha D_{\beta-\alpha})$ (where $D_0$ denotes the differentiation~$D$). If  $\Re(\alpha)\ge0$, $\Im(\alpha)>0$ and $\beta=\overline\alpha$, then $Z+\frac\sigma\alpha=UV$ and $U=(Z+\frac\sigma\alpha)(1-\alpha D_{\beta-\alpha})$.  In either case, the operator $1-\alpha D_{\beta-\alpha}$ is a self-homeomorhism of $\mathcal E_{\mathrm{min}}^1(\mathbb C)$ and $V=(1-\alpha D_{\beta-\alpha})^{-1}$. 
\end{proposition}

The following proposition provides explicit formulas for the action of the operator  $V=(1-\alpha D_{\beta-\alpha})^{-1}$ in $\mathcal E_{\mathrm{min}}^1(\mathbb C)$.

\begin{proposition}\label{njfdyde6} Let $f\in\mathcal E_{\mathrm{min}}^1(\mathbb C)$. If $\alpha=\beta>0$, 
then
\begin{equation}\label{vcsraeE33e}
(Vf)(z)=\int_{\R_+}f(z+x)\,\mu_{\alpha,\alpha,\eta}(dx),\quad z\in\mathbb C,
\end{equation}
if $\alpha>\beta>0$, then
\begin{equation}\label{fzGCXDTSTweq}
(Vf)(z)=\int_{(\alpha-\beta)\mathbb N_0}\big(f(z+x)\alpha/\beta-f(z)(\alpha-\beta)/\beta\big)\,\mu_{\alpha,\beta,\eta}(dx),\quad z\in\mathbb C,
\end{equation}
and if $\Re(\alpha)\ge0$, $\Im(\alpha)>0$ and $\beta=\overline\alpha$, then 
\begin{equation}\label{cxesawra4ea}
(Vf)(z)=\int_{\R}\big(f(z+x+\beta)\alpha/\beta-f(z)(\alpha-\beta)/\beta\big)\,\mu_{\alpha,\beta,\eta}(dx),\quad z\in\mathbb C.
\end{equation} 
\end{proposition}

\section{Proofs}\label{vftds5a453q}
\subsection{Proof of Lemmas~\ref{yufd6e64w5}--\ref{xdesrea5qw3} }

\subsubsection{{The case of the gamma distribution and the negative binomial distribution}}
\label{vcrts5eywy5}

First, we will prove Lemmas~\ref{yufd6e64w5}--\ref{xdesrea5qw3} in the case $\alpha\ge\beta>0$. We divide the proof into several steps.

{\it Step 1}. By Corollary~\ref{vctrsw645bvf}, Proposition~\ref{f7iri7}, and formula~\eqref{cfgdyjk} in the Appendix, we get, for $z\in(-\sigma/\alpha,+\infty)$,  
\begin{align}
(\mathcal Sx^n)(z)&=(\mathcal S\mathcal S^{-1}\mathcal R^n \mathcal S1)(z)= 
(\mathcal R^n1)(z)= \sum_{k=1}^{n} (\alpha-\beta)^{n-k}  S(n,k)\,(z+\sigma/\alpha \mid -\beta)_k
\notag\\
&=\int_{X_{\alpha,\beta}} x^n\,\mu_{\alpha,\beta,\sigma+\alpha z}(dx).\label{cxdw3rvcrdr}
\end{align}
Hence, for each polynomial $p\in\mathcal P(\mathbb C)$ and $z\in(-\sigma/\alpha,+\infty)$,
\begin{equation}\label{vfye67i}
(\mathcal Sp)(z)=\int_{X_{\alpha,\beta}}p\,d\mu_{\alpha,\beta,\sigma+\alpha z}.
\end{equation}
Note that $(\mathcal Sp)(z)$ can be extended to an entire function of $z\in\mathbb C$.

\begin{lemma}\label{bvgdseaw}
Let $p\in\mathcal P(\mathbb C)$. Then the function
$\mathcal D_{\alpha,\beta,\sigma}\ni z\mapsto\int_{X_{\alpha,\beta}}p\,d\mu_{\alpha,\beta,\sigma+\alpha z}$
is well-defined and analytic.
\end{lemma}

\begin{proof} Let $\alpha=\beta$. For $z\in\mathcal D_{\alpha,\beta,\sigma}$, we have $\Re(\sigma+\alpha z)>\sigma/2$. Hence, it is sufficient to prove that, for each $k\in\mathbb N_0$, the function
$$\zeta\mapsto\int_0^\infty x^k\,\mu_{\alpha,\alpha,\zeta}(dx)=\frac{1}{\Gamma (\frac{\zeta}\eta)} \,  \alpha^{-\frac{\zeta}\eta}\int_0^\infty  x^{k-1+ \frac{\zeta}\eta} \, e^{- \frac{x}{\alpha}} \, dx$$
is well-defined and analytic on the domain $\{\zeta\in\mathbb C\mid\Re(\zeta)>0\}$. To this end, it is sufficient to check the analyticity of the function 
\begin{equation}\label{vcyrdu}
\zeta\mapsto\int_0^\infty  x^{k-1+ \frac{\zeta}\eta} \, e^{- \frac{x}{\alpha}} \, dx.\end{equation}
We have
$$\int_0^\infty  \big|x^{k-1+ \frac{\zeta}\eta} \, e^{- \frac{x}{\alpha}} \big|\, dx=\int_0^\infty  x^{k-1+ \frac{\Re(\zeta)}\eta} \, e^{- \frac{x}{\alpha}} \, dx<\infty,$$
hence the function in \eqref{vcyrdu} is well defined. Furthermore, 
\begin{equation}\label{vyr6eufd}
\bigg|\frac{d}{d\zeta}\,x^{k-1+ \frac{\zeta}\eta} \, e^{- \frac{x}{\alpha}} \bigg|=
\frac{1}{\eta}\,|\log(x)|\,x^{k-1+ \frac{\Re(\zeta)}\eta} \, e^{- \frac{x}{\alpha}}.
\end{equation}
Note that $\log(x)\le x$ for $x\ge1$ and, for each  $\varepsilon>0$ there exists $C_1>0$ such that  $|\log(x)|\le C_1x^{-\varepsilon}$ for $x\in(0,1)$. Hence,  formula \eqref{vyr6eufd} easily implies that the function in~\eqref{vcyrdu} is indeed analytic on $\{\zeta\in\mathbb C\mid\Re(\zeta)>0\}$.

Next, let $\alpha>\beta$. It is sufficient to prove that, for each $k\in\mathbb N_0$, the function
\begin{equation}\label{vctesw5yu}
\zeta\mapsto \int_{(\alpha-\beta)\mathbb N_0} x^k\,\mu_{\alpha,\beta,\zeta}(dx)=\left( 1-\frac{\beta}{\alpha} \right)^{\frac{\zeta}{\eta}} \sum_{n=0}^\infty \left( \frac{\beta}{\alpha} \right)^n \frac{1}{n!} \left( \frac{\zeta}{\eta} \right)^{(n)} \big((\alpha-\beta)n\big)^k
\end{equation}
is entire. We have
$$\sum_{n=0}^\infty \left( \frac{\beta}{\alpha} \right)^n \frac{1}{n!}\,\bigg|\left( \frac{\zeta}{\eta} \right)^{(n)}\bigg|\, \big((\alpha-\beta)n\big)^k\le \sum_{n=0}^\infty \left( \frac{\beta}{\alpha} \right)^n \frac{1}{n!} \left( \frac{|\zeta|}{\eta} \right)^{(n)} \big((\alpha-\beta)n\big)^k<\infty,$$
because each monomial $x^k$ is integrable with respect to the negative binomial distribution~$\mu_{\alpha, \beta, \lvert \zeta\rvert}$. Hence, the series in \eqref{vctesw5yu} converges uniformly on compact sets in $\mathbb C$, which implies that the function in  \eqref{vctesw5yu} is entire. 
\end{proof}

Now formula \eqref{vfye67i}, Lemma~\ref{bvgdseaw}, and the identity theorem for analytic functions imply
\begin{equation}\label{gcfsrthe5u86}
(\mathcal S p)(z)= \int_{X_{\alpha,\beta}} p\,d\mu_{\alpha,\beta,\sigma+\alpha z},\quad p\in\mathcal P(\mathbb C),\ z\in\mathcal D_{\alpha,\beta,\sigma}.
\end{equation}

{\it Step 2}. Let $\alpha=\beta$, let $f\in L^2(\R_+,\mu_{\alpha,\alpha,\sigma})$ and let $-\frac{\sigma}{2\alpha} < \delta < \Delta < +\infty$. Then, for each $z \in \mathbb C$ with $\delta \leq \Re (z) \leq \Delta$, we have
\begin{align}
&\int_0^{\infty} \lvert f(x) \rvert \, \lvert x^{-1+\frac{\alpha z+\sigma}{\eta}} \rvert \, e^{- \frac{x}{\alpha}} \, dx \notag = \int_0^{\infty} \lvert f(x) \rvert \, x^{-1+\frac{\alpha \Re (z)+\sigma}{\eta}} \, e^{- \frac{x}{\alpha}} \, dx \notag \\
&\quad\leq \left( \int_0^{\infty} \lvert f(x) \rvert^2 \, x^{-1+\frac{\sigma}{\eta}} \, e^{- \frac{x}{\alpha}} \, dx \right)^\frac{1}{2} \left( \int_0^{\infty} x^{-1+\frac{\sigma+ 2\alpha\Re (z) }{\eta}} \, e^{- \frac{x}{\alpha}} \, dx \right)^{\frac{1}{2}}\notag\\
&\quad \le C_2 \, \lVert f \rVert_{L^2(\mu_{\alpha, \alpha, \sigma})} \label{76iry6u}
\end{align}
for a constant $C_2>0$ that depends on $\delta$ and $\Delta$. Now write $f(x)=\sum_{n=0}^\infty f_ns_n(x) $ and define, for $N\in\mathbb N$, $p_N(x)=\sum_{n=0}^Nf_ns_n(x)$. Formulas \eqref{gcfsrthe5u86} and  \eqref{76iry6u} imply that  $(\mathcal S p_N )(z)=\sum_{n=0}^N f_n(z\mid\alpha)_n $ converges uniformly on compact sets in $\mathcal D_{\alpha,\alpha,\sigma}$ to an analytic function and 
\begin{equation}\label{rs4qa4q}
(\mathcal Sf)(z)=\sum_{n=0}^\infty f_n (z\mid\alpha)_n=\int_0^{\infty} f \, d \mu_{\alpha, \alpha, \alpha z + \sigma}.\end{equation}

{\it Step 3.} Let $\alpha>\beta$ and  let $f\in L^2((\alpha-\beta)\mathbb N_0,\mu_{\alpha,\beta,\sigma})$.  We have
\begin{align}
&\sum_{n=0}^\infty \ \lvert f( (\alpha - \beta)n ) \rvert \left( \frac{\beta}{\alpha} \right)^n \frac{1}{n!} \, \bigg\lvert \left( \frac{\sigma + \alpha z}{\eta} \right)^{(n)} \bigg\rvert	\notag \\
&\qquad \leq \left( 1 - \frac{\beta}{\alpha} \right)^{-\frac{\sigma}{\eta}} {\lVert f \rVert}_{ L^2 (\mu_{\alpha, \beta, \sigma})}  \bigg( \sum_{n=0}^\infty \left( \frac{\beta}{\alpha} \right)^n \frac{\left[ \left( \frac{\sigma + \alpha \lvert z \rvert}{\eta} \right)^{(n)} \right]^2}{n! \left( \frac{\sigma}{\eta} \right)^{(n)}} \bigg)^{\frac{1}{2}} . \label{gf45gp}
\end{align}

\begin{lemma} \label{5t4g8} For any $\, a_1 > a_2 > 0$ and $\, 0 < q < 1$, 
$$ \sum_{n=0}^\infty q^n \frac{\left[(a_1)^{(n)}\right]^2}{n!\, (a_2)^{(n)}} < \infty. $$
\end{lemma}

\begin{proof}
It follows from the construction of a negative binomial distribution that, for each $q\in(0,1)$ and $a_1 > 0$, 
$\sum_{n=0}^\infty q^n \, \frac{1}{n!} \, (a_1)^{(n)} < \infty$.
Therefore, for each $\varepsilon>0$, we have
$(a_1)^{(n)} \leq C_3  (1+\varepsilon)^n\,n!$\,, 
 where the constant $ C_3 > 0 $ depends only on $a_1$ and $\varepsilon$. Next, for any $a_2 > 0$,  $(a_2)^{(n)}\ge a_2  (n-1)!$\,. Therefore, for any $a_1>a_2>0$ and $\varepsilon>0$,
$
(a_1)^{(n)}/(a_2)^{(n)} \leq C_4 (1+\varepsilon)^n
$,	
 where $C_4 > 0 $ depends on $a_1$, $a_2$ and $\varepsilon$. Hence, 
$$\sum_{n=0}^\infty q^n \, \frac{[(a_1)^{(n)}]^2}{n! \, (a_2)^{(n)}} \leq C_3 C_4 
\sum_{n=0}^\infty  \big(q(1+\varepsilon)^2\big)^n < \infty,$$
if we choose $\varepsilon >0$ such that $(1+\varepsilon)^2< 1/q$. 	
\end{proof}

Using estimate \eqref{gf45gp} and Lemma~\ref{5t4g8}, we now show similarly to Step 2 that, for $f(x)=\sum_{n=0}^\infty f_n s_n(x)\in L^2((\alpha-\beta)\mathbb N_0,\mu_{\alpha,\beta,\sigma})$,  
\begin{equation}\label{vctrst5rs4qa4q}
(\mathcal Sf)(z)=\sum_{n=0}^\infty f_n (z\mid\beta)_n=\int_0^{\infty} f \, d \mu_{\alpha, \beta, \alpha z + \sigma}, \end{equation}
the series in \eqref{vctrst5rs4qa4q} converges uniformly on compact sets in $\mathbb C$, and hence, $(\mathcal Sf)(z)$ is an entire function. 

Thus, Lemma~\ref{xfsxersarewa4w} (i) is proven. 

{\it Step 4.} To finish the proof of Lemma~\ref{yufd6e64w5}, we only need to prove formula \eqref{i78yu}. In fact, the first equality in \eqref{i78yu} is an immediate consequence of the fact that $D_\beta$ is the lowering operator for the polynomial sequence $\big((z\mid\beta)_n\big)_{n=0}^\infty$. 
The second equality in \eqref{i78yu} is a well-known identity for the $n$th difference operator, see e.g.\ formula~(6.2) in \cite{14-Quaintance}. 

{\it Step 5.} Let $z\in\mathcal D_{\alpha,\beta,\sigma}$. It follows from Steps 3 and 4 that  there exists a constant $C_5>0$ such that, for all $f\in L^2(X_{\alpha,\beta},\mu_{\alpha,\beta,\sigma})$, we have $|(\mathcal Sf)(z)|\le C_5\|f\|_{L^2(\mu_{\alpha,\beta,\sigma})}$. Hence, by the Riesz representation theorem, there exists $\mathcal K_z\in L^2(X_{\alpha,\beta},\mu_{\alpha,\beta,\sigma})$ such that 
\begin{equation}\label{vfye6wqa}(\mathcal Sf)(z)=\int_{X_{\alpha,\beta}}f(x)\mathcal K_z(x)\,\mu_{\alpha,\beta,\sigma}(dx)\quad\text{for all }f\in L^2(X_{\alpha,\beta},\mu_{\alpha,\beta,\sigma}).
\end{equation}
By \eqref{7ir7yj} and \eqref{vfye6wqa}, we conclude that $\mathcal K_z(x)=\mathcal E(x,z)$, where $\mathcal E(x,z)$ is given by \eqref{5yrthfg}, and $\mathcal K_z(x)$ is the Radon--Nykodim derivative $\frac{d\mu_{\alpha,\beta,\alpha z+\sigma}}{d\mu_{\alpha,\beta,\sigma}}(x)$. This easily implies Lemma~\ref{vcydyrtds6y}.

{\it Step 6.} In view of Proposition~\ref{f7iri7}, to prove Lemma~\ref{xdesrea5qw3}, we only need to check that the operator $\mathcal R$ with domain $\mathcal P(\mathbb C)$ is essentially self-adjoint in ${\mathcal F}_{\alpha, \beta, \sigma} $. But this can be easily shown by using Nelson's analytic vector criterium, see e.g.\ \cite[Section~X.6]{ReedSimon}. 

\subsubsection{The case of the Meixner distribution}

Now we consider the case $\Re(\alpha)\ge0$, $\Im(\alpha)>0$, $\beta=\overline\alpha$. We again divide the proof into several steps.

{\it Step 1}. Let $K>1$ be fixed. We state that there exists a constant $C_1>0$ such that
\begin{equation}\label{crstrs}
|\Gamma(ix+y)|\le C_1 \exp \left(-\frac{\pi}{2} \,\lvert x \rvert \right)(1+|x|)^{K},\quad x\in\R,\ y\in[1/K,K]. 
\end{equation}

Indeed, by e.g.~\cite[p.~15]{Gamma-Lebedev}, the following asymptotic formula holds, for all $x,y\in\R$, $-ix-y\not\in\mathbb N_0$:
	\begin{equation*}
		\lvert \Gamma(ix+y) \rvert = \sqrt{2\pi} \,  \exp \left(-\frac{\pi}{2} \lvert x \rvert \right) \, {\lvert x \rvert}^{ y -\frac{1}{2}} \, \left( 1+ E(x,y)\right),
	\end{equation*}
where the function $E(x,y)$ satisfies, for   each fixed $R>0$, 
$$\lim_{|x|\to\infty}\sup_{y\in[-R,R]}|E(x,y)|=0.$$
 Hence, formula~\eqref{crstrs} easily follows if we take into account that the function
	$$\R\times[1/K,K]\ni(x,y)\mapsto |\Gamma(ix+y)|\in\R$$
is continuous, hence bounded on $[-L,L]\times[1/K,K]$ for each $L>0$.

{\it Step 2}. Recall the domain $\mathfrak D_{\alpha,\beta}$ defined in Section~\ref{crtds6ewu65e3}. We state that, for each $p\in\mathcal P(\mathbb C)$, the function
$\mathfrak D_{\alpha,\beta}\ni \zeta\mapsto\int_\R p(x)\,\mu_{\alpha,\beta,\zeta}(dx)\in\mathbb C$
is well-defined and analytic.

Indeed, since $\operatorname{Arg}(\alpha) \in (0, \pi/2]$, we have $\cos\big(\frac\pi2 - \operatorname{Arg}(\alpha)\big) > 0$. Therefore, the function
$$\zeta\mapsto C_{\alpha,\beta,\zeta}=\frac{\big(2\cos\big(\frac\pi2 - \operatorname{Arg}(\alpha)\big)\big)^{\frac{\zeta}{\eta }}}{4\Im(\alpha)\pi \, \Gamma(\frac{\zeta}{\eta })} \, \exp\bigg( \frac{(\frac\pi2-\operatorname{Arg}(\alpha))\zeta\Re(\alpha)}{\Im(\alpha)\eta } \bigg)\in\mathbb C$$
is analytic on the domain $\{\zeta\in\mathbb C\mid\Re(\zeta)>0\}$. Hence, it is sufficient to prove that, for each $n\in\mathbb N_0$, the following function is well-defined and analytic:
\begin{equation}\label{vcyrsd5ws53q5q}
\mathfrak D_{\alpha,\beta}\ni \zeta\mapsto\int_\R x^n\exp\big( (\pi/2 - \operatorname{Arg}(\alpha))x/\Im(\alpha) \big)g_{\alpha,\beta,\zeta}(x)\,dx\in\mathbb C,\end{equation}
where
\begin{align}
g_{\alpha,\beta,\zeta}(x)&= \Gamma \left( \frac{ix}{2\Im(\alpha)} + \frac{i\zeta}{2\alpha\Im(\alpha)} \right)
	 \Gamma \left( -\frac{ix}{2\Im(\alpha)} - \frac{i\zeta}{2\beta\Im(\alpha)} \right)\label{vctsreser5sw}\\
&=\Gamma \big( d_1(\zeta) + i (l_1(\zeta)+x/(2 \Im(\alpha))  \big)\Gamma \big( d_2(\zeta) + i (l_2(\zeta)-x/(2 \Im(\alpha))  \big).\notag	 
	 \end{align}
Here
\begin{align}
		d_1(\zeta) &= \frac{\Re(\zeta) \Im(\alpha) - \Im(\zeta) \Re(\alpha)}{2 \eta  \Im(\alpha)},\quad 		l_1(\zeta) = \frac{\Re(\zeta) \Re(\alpha) + \Im(\zeta) \Im(\alpha)}{2 \eta  \Im(\alpha)},\notag\\
		d_2(\zeta) &= \frac{\Re(\zeta) \Im(\alpha) + \Im(\zeta) \Re(\alpha)}{2 \eta  \Im(\alpha)},	\quad
		l_2(\zeta) = \frac{-\Re(\zeta) \Re(\alpha)+ \Im(\zeta) \Im(\alpha)}{2 \eta  \Im(\alpha)} 	.\label{ctrsese5s}
	\end{align}
For each $\zeta\in \mathfrak D_{\alpha,\beta}$, we have $d_1(\zeta)>0$ and $d_2(\zeta)>0$. Therefore, for a fixed $x\in\R$, the function $\mathfrak D_{\alpha,\beta}\ni\zeta\mapsto g_{\alpha,\beta,\zeta}(x)\in\mathbb C$ is analytic. 

Let $\zeta \in \mathfrak D_{\alpha, \beta}$ be fixed. For $R>0$, denote $B(\zeta,R)=\{z\in\mathbb C\mid |z-\zeta|\le R\}$. Choose $R>0$ such that $B(\zeta,R)\subset \mathfrak D_{\alpha,\beta}$. 
		To prove the differentiability of the map in~\eqref{vcyrsd5ws53q5q} at point $\zeta$, it is sufficient to prove that 
	\begin{equation}	\label{45kjlo90}
		\int_{\mathbb R} \lvert x \rvert^n \,\exp\big( (\pi/2 - \operatorname{Arg}(\alpha))|x|/\Im(\alpha) \big) \sup_{z \in B(\zeta, R/2)}  \bigg\lvert  \frac{\partial}{\partial z} \, g_{\alpha, \beta, z} (x) \bigg\rvert \, dx < \infty.
	\end{equation}
(We  used the inequality  $\pi/2 - \operatorname{Arg}(\alpha)\ge0$.)
By Cauchy's integral formula,  \eqref{45kjlo90} would follow from 
\begin{equation}	\label{drtsdrses45kjlo90}
		\int_{\mathbb R} \lvert x \rvert^n \,\exp\big( (\pi/2 - \operatorname{Arg}(\alpha))|x|/\Im(\alpha) \big) \sup_{z \in B(\zeta, R)}  |  g_{\alpha, \beta, z} (x) \rvert \, dx < \infty.
	\end{equation}

Choose $K>0$ such that, for all $z \in B(\zeta, R$), both  $d_1(z)$ and $d_2(z)$ belong to $[\frac1K,K]$.
	Denote	$L_1 = \max_{z \in B(\zeta, R)} \lvert l_1 (z) \rvert$ and $L_2 = \max_{z \in B(\zeta, R)} \lvert l_2 (z) \rvert$. Then, by \eqref{crstrs}, there exists a constant $C_2>0$ such that
\begin{equation}\label{cdtrscfdrs}
 \sup_{z \in B(\zeta, R)}  |  g_{\alpha, \beta, z} (x) |\le C_2\exp\bigg(-\frac{\pi\,|x|}{2\Im(\alpha)}\bigg)(1+|x|)^{2K}.
 \end{equation}
Hence, the integral in \eqref{45kjlo90} is bounded by the following integral
$$ C_2\int_{\mathbb R} \lvert x \rvert^n\exp
\big(-\operatorname{Arg}(\alpha)|x|/\Im(\alpha)\big)(1+|x|)^{2K}<\infty, $$
where we used that $\operatorname{Arg}(\alpha)>0$.

{\it Step 3}. Let $(p_n)_{n=0}^\infty$ be the monic polynomial sequence over $\mathbb C$ satisfying the recurrence formula~\eqref{vtyd6yed}. Thus, for $x\in\R$, we have $s_n(x)=p_n(x+\sigma/\alpha)$. Let $\mathcal I$ be the linear operator in $\mathcal P(\mathbb C)$ as defined in Proposition~\ref{f7iri7}. Define  $\widetilde {\mathcal S}=\mathcal I^{-1}$. Thus, $(\widetilde {\mathcal S}p_n)(z)=(z\mid\beta)_n$. Similarly to \eqref{cxdw3rvcrdr}, we conclude from Corollary~\ref{vctrsw645bvf} and Proposition~\ref{f7iri7} that 
\begin{equation}\label{cxterawa44}
(\widetilde{\mathcal S}\zeta^n)(z)=\sum_{k=1}^{n} (\alpha-\beta)^{n-k}  S(n,k)\,(z+\sigma/\alpha \mid -\beta)_k\,,\quad n\in\mathbb N.
\end{equation}
On the other hand, for each $r\in(-\sigma/\eta,+\infty)$ and $z=\beta r$, we have $\sigma+\alpha z=\sigma+\eta r\in(0,\infty)$. Hence, by  \eqref{vcdtre54w5uw5xs},
\begin{equation}\label{xts5ew5}
\int_{\mathbb R} (x +z+ \sigma/\alpha)^n \, \mu_{\alpha, \beta, \sigma+\alpha z}(dx) = \sum_{k=1}^n (\alpha-\beta)^{n-k}  S(n,k)  (z+\sigma/\alpha \mid -\beta)_k\,,\quad n\in\mathbb N . 
\end{equation}
By \eqref{cxterawa44} and \eqref{xts5ew5}, we have, for each $p\in\mathcal P(\mathbb C)$,
\begin{equation}\label{vydt6u}
(\widetilde{\mathcal S}p)(z)=\int_{\mathbb R} p(x +z+ \sigma/\alpha) \, \mu_{\alpha, \beta, \sigma+\alpha z}(dx),\quad z=\beta r,\ r\in(-\sigma/\eta,+\infty). 
\end{equation}
Setting $p=p_n$ into \eqref{vydt6u} gives
\begin{align}
&\int_{\mathbb R} s_n(x +z) \, \mu_{\alpha, \beta, \sigma+\alpha z}(dx)=
\int_{\mathbb R} p_n(x +z+ \sigma/\alpha) \, \mu_{\alpha, \beta, \sigma+\alpha z}(dx)\notag\\
&\quad =(\widetilde{\mathcal S}p_n)(z)=(z\mid\beta)_n=(\mathcal Ss_n)(z),\quad z=\beta r,\ r\in(-\sigma/\eta,+\infty),\ n\in\mathbb N_0.\notag
\end{align}
Therefore, for each  $p\in\mathcal P(\mathbb C)$,
\begin{equation}\label{yd6e6}
(\mathcal Sp)(z)=\int_{\mathbb R} p(x +z) \, \mu_{\alpha, \beta, \sigma+\alpha z}(dx),\quad z=\beta r,\ r\in(-\sigma/\eta,+\infty).
\end{equation}

Recall the open domain $\Psi_{\alpha,\beta,\sigma}$ defined by \eqref{ctrs5wu65cdrsdr}.  Obviously, 
$$\{z=\beta r\mid r\in(-\sigma/\eta,+\infty)\}\subset \Psi_{\alpha,\beta,\sigma}.$$ 
It follows from Step 2 that, for each $p\in\mathcal P(\mathbb C)$, the function
$$\Psi_{\alpha,\beta,\sigma}\ni z\mapsto \int_{\mathbb R} p(x +z) \, \mu_{\alpha, \beta, \sigma+\alpha z}(dx)\in\mathbb C$$
is analytic. On the other hand, $(\mathcal Sp)(z)$ is an entire function. Hence, by \eqref{yd6e6} and the identity theorem for analytic functions,
\begin{equation}\label{xse5wyu}
(\mathcal Sp)(z)=\int_{\mathbb R} p(x +z) \, \mu_{\alpha, \beta, \sigma+\alpha z}(dx),\quad z\in \Psi_{\alpha,\beta,\sigma}.
\end{equation}

{\it Step 4}. A direct calculation shows that, if $\Re(\alpha)=0$ then $\R\subset \Psi_{\alpha,\beta,\sigma}$, and if $\Re(\alpha)>0$ then $(-\sigma/(2\Re(\alpha)),\infty)\subset \Psi_{\alpha,\beta,\sigma}$. Below we will use the notation $-\sigma/(2\Re(\alpha))$ even if $\Re(\alpha)=0$, meaning that  $-\sigma/(2\Re(\alpha))=-\infty$. Hence, by \eqref{xse5wyu},
\begin{equation}\label{vfrtyde65e564}
(\mathcal Sp)(z)=\int_{\mathbb R} p(x +z) \, \mu_{\alpha, \beta, \sigma+\alpha z}(dx),\quad z\in 
(-\sigma/(2\Re(\alpha)),\infty).\end{equation}
The change of variable $x'=x+z$ in the integral in \eqref{vfrtyde65e564} implies
\begin{equation}\label{vcrdstrsr}
 (\mathcal S p)(z) = \int_{\mathbb R} p(x)\, \mathcal G(x,z) \, dx,\quad  z\in 
(-\sigma/(2\Re(\alpha)),\infty),
\end{equation}
where for $x\in\R$ and $z\in (-\sigma/(2\Re(\alpha)),\infty)$,
\begin{align}
		&\mathcal {G} (x,z) = \frac{\big(2\cos\big(\frac\pi2 - \operatorname{Arg}(\alpha)\big)\big)^{\frac{\alpha z + \sigma}{\eta}}}{4\Im(\alpha)\pi  \Gamma(\frac{\alpha z + \sigma}{\eta})}\,    \exp\left( \frac{(\frac\pi2-\operatorname{Arg}(\alpha))(\alpha z + \sigma)  \Re(\alpha)}{\Im(\alpha)  \eta} \right) \notag	\\
							    &\quad\times \exp\left( \frac{(\frac\pi2 - \operatorname{Arg}(\alpha)) (x -z)}{\Im(\alpha)} \right) \Gamma \left( \frac{i x}{2  \Im(\alpha)} + \frac{i \sigma \beta}{2 \eta  \Im(\alpha)} \right)  \Gamma \left( \frac{-ix}{2 \Im(\alpha)} - \frac{i \sigma \alpha}{2 \eta  \Im(\alpha)} + \frac{z \alpha}{\eta }  \right).\label{drs5ws5w5}	
	\end{align}
(Note that the real part of the argument of each of the gamma functions in \eqref{drs5ws5w5} is positive.)

 Recall the domain $\mathcal D_{\alpha,\beta,\sigma}$ defined by \eqref{tedrer4we4}. It is straightforward to see that, for each fixed $x\in\R$, the function $\mathcal G(x,\cdot)$ admits a unique extension to an analytic function on $\mathcal D_{\alpha,\beta,\sigma}$, and this extension is still given by formula \eqref{drs5ws5w5}. Similarly to Step~3, we show that, for each $p\in\mathcal P(\mathbb C)$, the function
$\mathcal D_{\alpha,\beta,\sigma}\ni z\mapsto  \int_{\mathbb R} p(x)\, \mathcal G(x,z) \, dx\in\mathbb C$
is analytic. Therefore, formula \eqref{vcrdstrsr} implies
\begin{equation}\label{vreseseswtrsr}
 (\mathcal S p)(z) = \int_{\mathbb R} p(x)\, \mathcal G(x,z) \, dx,\quad  z\in 
\mathcal D_{\alpha,\beta,\sigma}.
\end{equation}

{\it Step 5}. Let $z_0\in \mathcal D_{\alpha,\beta,\sigma}$. Choose $R>0$ such that the closed ball $B(z_0,R)$ is a subset of  $\mathcal D_{\alpha,\beta,\sigma}$ We state that there exists a constant $C_3>0$ such that, for all $f \in L^2 (\mathbb R, \mu_{\alpha, \beta, \sigma})$,
	\begin{equation}
		\sup_{z \in B(z_0, R)} \ \int_{\mathbb R} \lvert f(x) \rvert \, \lvert \mathcal {G}(x,z) \rvert \, dx		\leq 
		C_3  \lVert f \rVert_{L^2 ( \mu_{\alpha, \beta, \sigma})}.\label{665rgf}
	\end{equation}

Indeed, we have, by the Cauchy inequality, for each $z\in B(z_0,R)$,
\begin{equation}
		\int_{\mathbb R} \lvert f(x) \rvert \, \lvert \mathcal {G}(x,z) \rvert \, dx	\le {\lVert f \rVert}_{L^2( \mu_{\alpha, \beta, \sigma})} \left( \int_{\mathbb R} \frac{ \lvert \mathcal {G} (x,z) \rvert^2 }{G_{\alpha,\beta,\sigma}(x)} \, dx \right)^{\frac{1}{2}}. 	\label{ctrsw5rw}
			\end{equation} 
Here $G_{\alpha,\beta,\sigma}(x)$ is the density of the measure $\mu_{\alpha,\beta,\sigma}$ with respect to the Lebesgue measure:
$$ G_{\alpha,\beta,\sigma}(x)={C}_{\alpha, \beta, \sigma} \, \exp\left( (\pi/2 - \operatorname{Arg}(\alpha))x/\Im(\alpha) \right)g_{\alpha,\beta,\sigma}(x). $$
Using \eqref{sewa4tq}, \eqref{vctsreser5sw},  \eqref{drs5ws5w5}, and the equality $|\Gamma(\zeta)|=|\Gamma(\bar \zeta)|$, we find
\begin{align}
		&\frac{ |\mathcal {G}(x,z)| ^2 }{{G}_{\alpha, \beta, \sigma} (x)}  	 = \bigg|\frac{\big(2\cos\big(\frac\pi2 - \operatorname{Arg}(\alpha)\big)\big)^
		{\frac{2 \alpha z+\sigma}{\eta }}}{4 \Im(\alpha)\pi} \cdot \frac{\Gamma(\frac{\sigma}{\eta })}{\Gamma^2 \left( \frac{\alpha z + \sigma}{\eta} \right)} 	\notag\\
		&\quad\times \exp\left[ \frac{(\frac\pi2-\operatorname{Arg}(\alpha)) \, \Re(\alpha) \, (2 \alpha z + \sigma)}{\Im(\alpha)\eta } +\frac{(\frac\pi2 - \operatorname{Arg}(\alpha)) (x - 2z)}{\Im(\alpha)} \right] \notag 	\\
		&\quad\quad \times \Gamma^2 \left( -\frac{i x}{2\Im(\alpha)} - \frac{i \sigma \alpha}{2\eta  \, \Im(\alpha)} + \frac{z \alpha}{\eta} \right)\bigg|.\label{drdrdsrs4}
		\end{align}	
Thus, by \eqref{ctrsw5rw} and \eqref{drdrdsrs4}, to prove \eqref{665rgf}, it is sufficient to show that 			\begin{equation}
	\sup_{z \in B(z_0, R)} \, \int_{\mathbb R}  \exp\left[ \frac{(\frac\pi2 - \operatorname{Arg}(\alpha)) \,|x|}{\Im(\alpha)} \right] \cdot	\bigg|\Gamma \left( -\frac{i x}{2\Im(\alpha)} - \frac{i \sigma \alpha}{2\eta   \Im(\alpha)} + \frac{z \alpha}{\eta} \right) \bigg|^2	\, dx < \infty. 	\label{62dhj9o}
			\end{equation}
We note that, for each $z\in B(z_0,R)\subset \mathcal D_{\alpha,\beta,\sigma}$,
$$\Re\bigg(- \frac{i \sigma \alpha}{2\eta   \Im(\alpha)} + \frac{z \alpha}{\eta}\bigg)=\frac\sigma{2\eta}+\frac{\Re(\alpha z)}{\eta}>0. $$
Hence, there exists $K>1$ such that 
$$\Re\bigg(- \frac{i \sigma \alpha}{2\eta   \Im(\alpha)} + \frac{z \alpha}{\eta}\bigg)\in[1/K,K]\quad\forall z\in B(z_0,R).$$
Since $\Re(\alpha z)$ is bounded on $B(z_0,R)$, estimate \eqref{crstrs} easily implies \eqref{62dhj9o}. (Compare with Step~2.)

{\it Step 6}. Similarly to Step 2 in Subsection~\ref{vcrts5eywy5},  we conclude from \eqref{665rgf} that, for each $f(x)=\sum_{n=0}^\infty f_ns_n(x)\in L^2(\R,\mu_{\alpha,\beta,\sigma})$, the series
$\sum_{n=0}^\infty f_n(z\mid \beta)_n$ converges uniformly on compact sets in $\mathcal D_{\alpha,\beta,\sigma}$ to an analytic function and
\begin{equation}\label{vhvcygtdcyd}
(\mathcal Sf)(z)=\sum_{n=0}^\infty f_n(z\mid \beta)_=\int_\R f(x)\mathcal G(x,z)\,dx.
\end{equation}
Similarly to Step 4 in Subsection~\ref{vcrts5eywy5}, this proves Lemma~\ref{yufd6e64w5}.   Next, using  \eqref{sewa4tq}, \eqref{cdzserar4ftstew5}, \eqref{drs5ws5w5}, and \eqref{vhvcygtdcyd}, we find that
formulas~\eqref{4rggn9} and \eqref{fdtrs5ea43qq} hold for each $f\in L^2(\R,\mu_{\alpha,\beta,\sigma})$. Similarly to Step~5 in Subsection \ref{vcrts5eywy5}, we see that $\mathcal E(\cdot,z)\in L^2(\R,\mu_{\alpha,\beta,\sigma})$ for each $z\in\mathcal D_{\alpha,\beta,\sigma}$ and formula~\eqref{5yrthfg} holds. This proves Lemma~\ref{vcydyrtds6y}.

{\it Step 7}. Note that $((z\mid\beta)_n)_{n=0}^\infty$ is Sheffer sequence with generating function~\eqref{vcxresayw} in which $A(t)=0$ and $B(t)=\frac1\beta\,\log(1+\beta t)$. Therefore, by Theorem~\ref{bvgytdddsesea}, the linear operator $\mathcal S$, acting in $\mathcal P(\mathbb C)$ 
and satisfying $\mathcal Ss_n=(\cdot\mid\beta)_n$ ($n\in\mathbb N_0$), extends to a continuous linear operator in $\mathcal E_{\mathrm{min}}^1(\mathbb C)$. 

Let $f(z)=\sum_{n=0}^\infty f_ns_n(z)\in\mathcal E_{\mathrm{min}}^1(\mathbb C)$ and define $p_N(z)=\sum_{n=0}^N f_ns_n(z)\in\mathcal P(\mathbb C)$ ($N\in\mathbb N$). Then $p_N\to f$ and $\mathcal Sp_N\to\mathcal Sf$ in $\mathcal E_{\mathrm{min}}^1(\mathbb C)$. In particular, for each fixed $z\in\mathbb C$, we have $p_N(z)\to f(z)$ and $(\mathcal Sp_N)(z)\to (\mathcal Sf)(z)$ as $N\to\infty$.

In view of \eqref{xse5wyu}, to prove Lemma~\ref{xfsxersarewa4w} (ii), it is sufficient to show that, for each $\zeta\in \mathfrak D_{\alpha,\beta}$ and $z\in\mathbb C$, we have
$$\lim_{N\to\infty}\int_{\R}p_N(x+z)\,\mu_{\alpha,\beta,\zeta}(dx)= \int_{\R}f(x+z)\,\mu_{\alpha,\beta,\zeta}(dx),$$
which is equivalent to 
\begin{align}
&\lim_{N\to\infty}\int_{\R}p_N(x+z)\exp\big( (\pi/2 - \operatorname{Arg}(\alpha))x/\Im(\alpha) \big)g_{\alpha,\beta,\zeta}(x)\,dx\notag\\
&\qquad= \int_{\R}f(x+z)\exp\big( (\pi/2 - \operatorname{Arg}(\alpha))x/\Im(\alpha) \big)g_{\alpha,\beta,\zeta}(x)\,dx.\label{fxdvfdtrs}
\end{align}

It follows from \eqref{cdtrscfdrs} that, for each $\zeta\in \mathfrak D_{\alpha,\beta}$, there exist constants $C_4>0$ and $K>0$ such that
\begin{equation}\label{vyd6}
|  g_{\alpha, \beta, \zeta} (x) |\le C_4\exp\bigg(-\frac{\pi\,|x|}{2\Im(\alpha)}\bigg)(1+|x|)^{2K}.
\end{equation} 
Since the sequence $(p_N)_{N=1}^\infty$ converges in  $\mathcal E_{\mathrm{min}}^1(\mathbb C)$, for each $t>0$, there exists a constant $C_t>0$ (depending on the fixed $z\in\mathbb C$) such that
\begin{equation}\label{csrere5aewaq}
\sup_{x\in\R}|p_N(x+z)|\le C_t\exp(t|x|).
\end{equation}
Choosing $t\in(0,\operatorname{Arg}(\alpha)/\Im(\alpha) )$, we conclude \eqref{fxdvfdtrs} from \eqref{vyd6}, \eqref{csrere5aewaq}, and the dominated convergence theorem. 

Thus, Lemma~\ref{xfsxersarewa4w} (ii) is proven. Finally, the proof of Lemma~\ref{xdesrea5qw3}  is similar to Step~6 in Subsection~\ref{vcrts5eywy5}.

\subsection{The remaining proofs}

\begin{proof}[Proof of Lemma~\ref{ftsxrearwa}]
We state that, for each  $z \in \mathbb C$,  
\begin{equation} \label{234hgjk}
		\int_{\mathbb N_0} (\xi)_n \, \pi_{z}(d\xi) = z^n ,\quad n\in\mathbb N.
							\end{equation}	
 For $z>0$, equality \eqref{234hgjk} is well-known. (To show it, one can use the equality 
 $\int_{\mathbb N_0} \xi^n \, \pi_{\sigma}(d\xi) = \sum_{k=1}^n S(n,k)\sigma^k$ and formula~\eqref{65yhu7}.)
As easily seen, the function 
$\mathbb C \ni z \mapsto  \int_{\mathbb N_0} (\xi)_n \, \pi_{z}(d\xi)$ is entire.
 Hence, formula \eqref{234hgjk} holds for all $z \in \mathbb C$ by the identity theorem for analytic functions.

Formula~\eqref{234hgjk} implies $\int_{\mathbb N_0} (\beta \xi\mid\beta)_n \, \pi_{\frac z\beta}(d\xi) = z^n$. Therefore, formula~\eqref{87.iur7.k} holds for $f\in\mathcal P(\mathbb C)$. 

Let $f(z) = \sum_{n=0}^\infty  f_n (z \mid \beta)_n \in {\mathcal F}_{\alpha, \beta, \sigma} $. Using \eqref{234hgjk}, we have, for $z\in\mathbb C$,
\begin{align}
& \sum_{k=0}^\infty |f(\beta k)|\,\frac1{k!}\,\bigg|\frac z\beta\bigg|^k\le 
 \sum_{k=0}^\infty \sum_{n=0}^\infty | f_n | \, \big| (\beta k \mid \beta)_n \big| \, \frac{1}{k!} \, \bigg|\frac z\beta\bigg|^k \notag\\
 &\quad= \sum_{n=0}^\infty | f_n | \, |\beta|^n\sum_{k=0}^\infty  (k)_n \, \frac{1}{k!} \, \bigg|\frac z\beta\bigg|^k 
 = \exp\big(|z|/\beta\big)\sum_{n=0}^\infty | f_n | \, |\beta|^n \int_{{\mathbb N}_0} (\xi)_n \, \pi_{\left|\frac{z}{\beta}\right|} (d\xi)\notag\\ 
&\quad = \exp\big(|z|/|\beta|\big)\sum_{n=0}^\infty | f_n | \, |z|^n\le
 \exp\big(|z|/|\beta|\big)
\bigg( \sum_{n=0}^\infty \frac{|z|^{2n}}{n!\, (\sigma \mid -\eta)_n} \bigg)^{1/2} {\lVert f \rVert}_{{\mathcal F}_{\alpha, \beta, \sigma} }. \notag
\end{align}
	Hence,  the integral on the right-hand side of formula \eqref{87.iur7.k} is well-defined and formula~\eqref{87.iur7.k} holds.	
\end{proof}

\begin{proof}[Proof of Theorem~\ref{vcfxsra5u6e756}] By Lemmas~\ref{vcydyrtds6y} and~\ref{ftsxrearwa}, we have, for each $f\in L^2(X_{\alpha,\beta},\mu_{\alpha,\beta,\sigma})$ and $z\in\mathbb C$,
\begin{equation}\label{dsawawaww}
(\mathbb Sf)(z)=\int_{\mathbb N_0}\int_{X_{\alpha,\beta}} f(x)\mathcal E(x,\beta \xi)\,\mu_{\alpha,\beta,\sigma}(dx)\,\pi_{\frac z\beta}(d\xi).\end{equation}
By \eqref{5yrthfg}, we have, for $x\in X_{\alpha,\beta}$ and $\xi\in\mathbb N_0$,
$$|\mathcal E(x,\beta\xi)| \le \sum_{n=0}^\infty \frac{|(\beta\xi \mid \beta)_n|}{n! \, (\sigma \mid - \eta)_n} \, |s_n (x)|=\sum_{n=0}^\infty \frac{|\beta|^n}{n! \, (\sigma \mid - \eta)_n} \, (\xi)_n|s_n (x)|.$$
Therefore, using  \eqref{234hgjk}, we obtain, for $f\in L^2(X_{\alpha,\beta},\mu_{\alpha,\beta,\sigma})$ and $z\in\mathbb C$,
\begin{align}
&\int_{\mathbb N_0}\int_{X_{\alpha,\beta}} |f(x)|\,|\mathcal E(x,\beta \xi)|\,\mu_{\alpha,\beta,\sigma}(dx)\,\pi_{|\frac z\beta|}(d\xi)\notag\\
&\quad\le \sum_{n=0}^\infty \frac{|\beta|^n}{n! \, (\sigma \mid - \eta)_n} \int_{X_{\alpha,\beta}}|f(x)s_n(x)| \bigg(\int_{\mathbb N_0}(\xi)_n\,\pi_{|\frac z\beta|}(d\xi)\bigg)\mu_{\alpha,\beta,\sigma}(dx)\notag\\
&\quad= \sum_{n=0}^\infty \frac{|z|^n}{n! \, (\sigma \mid - \eta)_n} \int_{X_{\alpha,\beta}}|f(x)s_n(x)| \,\mu_{\alpha,\beta,\sigma}(dx)\notag\\
&\quad\le \|f\|_{L^2(\mu_{\alpha,\beta,\sigma})}\sum_{n=0}^\infty \frac{|z|^n}{\sqrt{n! \, (\sigma \mid - \eta)_n}} <\infty.
\label{vrts5}\end{align}
Formulas \eqref{dsawawaww} and \eqref{vrts5} imply formula~\eqref{ctstesa5a5} in which
\begin{equation}\label{rsea54q54y}
\mathbb E(x,z)=\int_{\mathbb N_0}\mathcal E(x,\beta \xi)\,\pi_{\frac z\beta}(d\xi).
\end{equation}
Formulas \eqref{ctstesa5a5}, \eqref{rsea54q54y} and Lemma~\ref{vcydyrtds6y} imply formulas \eqref{cfsersara5s}--\eqref{xawawawefljkreiojfu}. For each $z\in\mathbb C$, the map $L^2(X_{\alpha,\beta},\mu_{\alpha,\beta,\sigma})\ni z\mapsto (\mathbb Sf)(z)\in\mathbb C$ is continuous, see e.g.\ \eqref{dsawawaww}, \eqref{vrts5}. Hence, $\mathbb E(\cdot,z)\in L^2(X_{\alpha,\beta},\mu_{\alpha,\beta,\sigma})$ for each $z\in\mathbb C$. Formula~\eqref{ydtyrdsj6ej} is then also obvious. 
\end{proof} 

\begin{proof}[Proof of Corollary \ref{vcsara5a54w45ss}]
The equality $(\partial^+p,q)_{L^2(\mu_{\alpha,\beta,\sigma})}=(p,A^-q)_{L^2(\mu_{\alpha,\beta,\sigma})}$ for $p,q\in\mathcal P(\mathbb C)$ follows from \eqref{545-980+polk}. Since the adjoint of the operator $A^-$ is densely defined, the operator $A^-$ is closable. For $N\in\mathbb N$, define $\mathbb E_N(x,z)=\sum_{n=0}^N \frac{z^n}{n!\,(\sigma\mid-\eta)_n} \, s_n (x)$. Then, for each $z\in\mathbb C$, 
$\mathbb E_N(\cdot,z)\to\mathbb E(\cdot,z)$ in $L^2(X_{\alpha,\beta},\mu_{\alpha,\beta,\sigma})$ as $N\to\infty$, and $A^-\mathbb E_N(\cdot,z)=z\mathbb E_{N-1}(\cdot,z)\to z\mathbb E(\cdot,z)$ in $L^2(X_{\alpha,\beta},\mu_{\alpha,\beta,\sigma})$ as $N\to\infty$. Hence, $\mathbb E(\cdot,z)$ belongs to the domain of $A^-$ and $A^-\mathbb E(\cdot,z)=z\mathbb E(\cdot,z)$. 
\end{proof} 

\begin{proof}[Proof of Theorem \ref{cdszrea5r4w35}]
 By Lemma~\ref{xfsxersarewa4w} (i) and Lemma~\ref{ftsxrearwa}, we have, for $f\in L^2(X_{\alpha,\beta},\mu_{\alpha,\beta,\sigma})$ and $z\in\mathbb C$:
\begin{equation}\label{cytrd6ed6}
(\mathbb Sf)(z)= \int_{\mathbb N_0}\int_{X_{\alpha,\beta}}f(x)\,\mu_{\alpha,\beta,\eta\xi+\sigma}(dx)\,\pi_{\frac z\beta}(d\xi).
\end{equation}
To conclude from \eqref{cytrd6ed6} that formulas \eqref{vr5y365re5e54}, \eqref{vcdrt5e} hold, it is sufficient to show that 
\begin{equation}\label{vcdtsftdtr}
\int_{\mathbb N_0}\int_{X_{\alpha,\beta}}|f(x)|\,\mu_{\alpha,\beta,\eta\xi+\sigma}(dx)\,\pi_{\frac {|z|}\beta}(d\xi)<\infty.\end{equation}
But this is immediate since $f\in L^2(X_{\alpha,\beta},\mu_{\alpha,\beta,\sigma})$ and the left-hand side of \eqref{vcdtsftdtr} is equal to $(\mathbb S|f|)(|z|)$. 
\end{proof}

\begin{proof}[Proof of Theorem \ref{xtxstesar5ewa54aqy}]
By Theorem~\ref{bvgytdddsesea}, the operator $\mathbb S$ acts continuously in $\mathcal E_{\mathrm{min}}^1(\mathbb C)$. Now   the theorem follows from Lemma~\ref{xfsxersarewa4w} (ii) and Lemma~\ref{ftsxrearwa} (note that $\mathcal E_{\mathrm{min}}^1(\mathbb C)$ can be naturally embedded into $\mathcal F_{\alpha,\beta,\sigma}$.) 
Indeed, the only fact that needs to be  checked is that, for each $\xi\in\mathbb N_0$, we have $\beta\xi\in\Psi_{\alpha,\beta,\sigma}$. But this is immediate since $\alpha\beta\xi+\sigma=\eta\xi+\sigma>0$ and so $\alpha\beta\xi+\sigma\in \mathfrak D_{\alpha,\beta}$.\end{proof}

\begin{proof}[Proof of Proposition \ref{cfxctdxtsdrytdsy}] The proposition follows immediately from Lemma~\ref{xdesrea5qw3}.
\end{proof}

\begin{proof}[Proof of Propositions~\ref{xszswcxdzwawawa} and \ref{njfdyde6}] We divide the proof into several steps.

{\it Step 1}. The operator $1-\alpha D_{\beta-\alpha}$ maps a monic polynomial sequence to a monic polynomial sequence. Hence, it is bijective as a map in $\mathcal P(\mathbb C)$. 

The equality 
\begin{equation}\label{dxsrsr}
V=\alpha\partial^-+1=(1-\alpha D_{\beta-\alpha})^{-1}\quad\text{on $\mathcal P(\mathbb C)$}
\end{equation} easily follows from umbral calculus. Indeed, $\partial^-=B(D)$, where $D$ is the differentiation operator and the function $B$ is as in formula~\eqref{vcxresayw}, see e.g.\ \cite[Section~4.4]{11-Kung}.  By \cite{13-Meixner}, if $\alpha\ne\beta$,  we have 
\begin{equation}\label{ytdsrea4}
B(t)=\frac{e^{(\beta-\alpha)t}-1}{\beta-\alpha e^{(\beta-\alpha)t}}=\frac{\frac{e^{(\beta - \alpha)t}-1}{\beta - \alpha}}{1- \alpha \frac{e^{(\beta - \alpha)t}-1}{\beta - \alpha}}.
\end{equation}
  By Boole's formula (e.g.\ \cite[Section~4.3.1]{11-Kung}), for $h\in\mathbb C$, the $h$-derivative has the representation $D_h=\frac{e^{hD}-1}h$. Hence, by \eqref{ytdsrea4},
$\partial^-= D_{\beta-\alpha}(1- \alpha D_{\beta-\alpha})^{-1}$, which implies \eqref{dxsrsr}. In the case $\alpha=\beta$, we have $B(t)=\frac t{1-\alpha t}$, which similarly implies  \eqref{dxsrsr}.

Recall that, in the case $\alpha\ge\beta>0$, we have  $Z=UV$ on $\mathcal P(\mathbb C)$, and in the case  $\Re(\alpha)\ge0$, $\Im(\alpha)>0$, $\beta=\overline \alpha$, we have $Z+\frac\sigma\alpha=UV$ on $\mathcal P(\mathbb C)$. Since $V^{-1}=1-\alpha D_{\beta-\alpha}$, this immediately implies that $U=Z(1-\alpha D_{\beta-\alpha})$ in the former case, and $U=(Z+\frac\sigma\alpha)(1-\alpha D_{\beta-\alpha})$ in the latter case. 

{\it Step 2}. Similarly to Step~1, we easily find that  $\partial^-=D_{\alpha-\beta}(1-\beta D_{\alpha-\beta})^{-1}=\sum_{k=0}^\infty\beta^k D_{\alpha-\beta}^{k+1}$. 
 Since $D_{\alpha-\beta}$ is the lowering operator for the polynomial sequence\linebreak $((z\mid\alpha-\beta)_n)_{n=0}^\infty$, we get  
\begin{equation}\label{cdtrs5t}
 \big(\partial^- (\cdot \mid \alpha - \beta)_n\big)(z)=\sum_{k=0}^{n-1} \frac{n!}{k!} \, \beta^{n-k-1} \, (z \mid \alpha - \beta)_k,\quad n\in\mathbb N .\end{equation}
 
 {\it Step 3}. We state that, when $\alpha\ge\beta>0$, 
 \begin{equation}	\label{y5rtgp0}
\int_{X_{\alpha,\beta}}  (x \mid \alpha - \beta)_n \, \mu_{\alpha, \beta, \sigma} (dx) = \beta^n ( \sigma/\eta)^{(n)},
\end{equation}
and when $\Re(\alpha)\ge0$, $\Im(\alpha)>0$, $\beta=\overline\alpha$,
\begin{equation}\label{vydste}
\int_{\mathbb R} (x + \sigma/\alpha\mid\alpha-\beta)_n \, \mu_{\alpha, \beta, \sigma}(dx) = \beta^n ( \sigma/\eta )^{(n)}.
\end{equation}

Note that, when $\alpha=\beta>0$, formula \eqref{y5rtgp0} is just \eqref{cfgdyjk}. We will prove formula \eqref{y5rtgp0} when $\alpha>\beta>0$, the proof of \eqref{vydste} being similar. By \eqref{65yhu7} and \eqref{cfgdyjk}, we obtain
\begin{align*}
&\int_{(\alpha - \beta) {\mathbb N}_0}  (x \mid \alpha - \beta)_n \, \mu_{\alpha, \beta, \sigma}
(dx)=(\alpha - \beta)^n\int_{(\alpha - \beta) {\mathbb N}_0}   \big( x/(\alpha-\beta)\big)_n \, \mu_{\alpha, \beta, \sigma}  (dx)	\\
&\quad = \sum_{k=1}^n s(n,k) (\alpha - \beta)^{n-k} \int_{(\alpha - \beta) {\mathbb N}_0}  x^k \, \mu_{\alpha, \beta, \sigma} (dx)	\\
&\quad = \sum_{k=1}^n s(n,k) (\alpha - \beta)^{n-k}  \sum_{i=1}^k (\alpha - \beta)^{k-i}  S(k,i)  ( \sigma/\alpha \mid -\beta)_i 	\\
&\quad = \sum_{i=1}^n (\alpha - \beta)^{n-i} (\sigma/\alpha \mid -\beta )_i \sum_{k=i}^n s(n,k)  S(k,i)  \\
&\quad = (\sigma/\alpha \mid -\beta )_n= \beta^n ( \sigma/\eta )^{(n)}.
\end{align*}

{\it Step 4}. Let $p\in\mathcal P(\mathbb C)$. We state that, if $\alpha\ge\beta>0$,
\begin{equation}\label{vtdtr6sa54aj}
(\partial^-p)(z)=\int_{X_{\alpha,\beta}}\big(p(z+x) -p(z)\big)\beta^{-1}\, \mu_{\alpha,\beta,\eta}(dx),
\end{equation}
and if $\Re(\alpha)\ge0$, $\Im(\alpha)>0$, $\beta=\overline\alpha$,
\begin{equation}\label{vtdtcfsxesea54aj}
(\partial^-p)(z)=\int_{X_{\alpha,\beta}}\big(p(z+x+\beta) -p(z)\big)\beta^{-1}\, \mu_{\alpha,\beta,\eta}(dx),
\end{equation}

To prove formula \eqref{vtdtr6sa54aj}, it is sufficient to show that it holds for $p(z)=(z\mid\alpha-\beta)_n$ ($n\in\mathbb N$). Then, by \eqref{cdtrs5t} and \eqref{y5rtgp0},
\begin{align*}
&\int_{X_{\alpha,\beta}}\big((z+x\mid\alpha-\beta)_n -(z\mid\alpha-\beta)_n\big)\beta^{-1}\, \mu_{\alpha,\beta,\eta}(dx)\\
&\quad=\sum_{k=0}^{n-1}\binom nk (z\mid\alpha-\beta)_k\,\beta^{-1}\int_{X_{\alpha,\beta}}(x\mid\alpha-\beta)_{n-k}\,\mu_{\alpha,\beta,\eta}(dx)\\
&\quad=\sum_{k=0}^{n-1}\binom nk (z\mid\alpha-\beta)_k\,\beta^{n-k-1}(n-k)!=(\partial^-p)(z).
\end{align*}
The proof of \eqref{vtdtcfsxesea54aj} is similar. We only need to note that $\eta/\alpha=\beta$. 

Since $V=\alpha\partial^-+1$, formulas \eqref{vtdtr6sa54aj}, \eqref{vtdtcfsxesea54aj} imply that formulas \eqref{vcsraeE33e}--\eqref{cxesawra4ea} hold for $f(z)=p(z)\in\mathcal P(\mathbb C)$.

{\it Step 5}.  Using Theorem~\ref{bvgytdddsesea}, one can easily show that the operators $\partial^+$, $\partial^-$, $Z$ and $D_{\beta-\alpha}$ admit a (unique) extension to continuous linear operators in $\mathcal E_{\mathrm{min}}^1(\mathbb C)$.  Hence, $U$ and $V$ also admit a continuous extension, $Z=UV$, respectively  $Z+\sigma/\alpha=UV$, and  $U=Z(1-\alpha D_{\beta-\alpha})$, respectively  $U=(Z+\sigma/\alpha)(1-\alpha D_{\beta-\alpha})$. 

Finally, using the definition of the space $\mathcal E_{\mathrm{min}}^1(\mathbb C)$, we easily see that the integrals on the right-hand side of formulas \eqref{vcsraeE33e}--\eqref{cxesawra4ea} are well defined for each $f\in \mathcal E_{\mathrm{min}}^1(\mathbb C)$, and furthermore, the right-hand side of each of the formulas \eqref{vcsraeE33e}--\eqref{cxesawra4ea} determines a continuous linear operator in $\mathcal E_{\mathrm{min}}^1(\mathbb C)$. Hence,  formulas \eqref{vcsraeE33e}--\eqref{cxesawra4ea} hold for $f\in \mathcal E_{\mathrm{min}}^1(\mathbb C)$. Since  $(1-\alpha D_{\beta-\alpha}) V=V(1-\alpha D_{\beta-\alpha})=1$ in $\mathcal E_{\mathrm{min}}^1(\mathbb C)$, the operator $1-\alpha D_{\beta-\alpha}$ is invertible in $\mathcal E_{\mathrm{min}}^1(\mathbb C)$ and $V=(1-\alpha D_{\beta-\alpha})^{-1}$. 
\end{proof}

\appendix
\renewcommand{\thesection}{A}
\setcounter{equation}{0}

\section*{Appendix A. Normal ordering in a class of generalized Weyl algebras and its connection to orthogonal Sheffer sequences}


We consider a special class of generalized Weyl algebras.
For  $a,b\in\mathbb C$, we are interested in the complex free algebra in two generators $\mathcal U$ and $\mathcal V$ satisfying the commutation relation $[\mathcal V,\mathcal U]=a\mathcal V+b$. 

Recall that the Stirling numbers of the first kind, $s(n,k)$, and of the second kind, $S(n,k)$, are defined as the coefficients of the expansions $ (z) _n = \sum_{k=1}^{n} s(n,k) \, z^k$ and  $ z^n = \sum_{k=1}^{n} S(n,k) (z)_k$, respectively. This definition immediately implies the orthogonality property of the Stirling numbers:
\begin{equation} \label{65yhu7} \sum_{k=i}^{n} S(n,k)s(k,i)  = \sum_{k=i}^{n} s(n,k)S(k,i)   = \delta_{n,i},\quad 1\le i\le n. \end{equation}

\begin{proposition} \label{UVtoN}
Assume that the generators $\mathcal U$, $\mathcal V$ satisfy $[\mathcal V,\mathcal U]=a\mathcal V+b$. Then, for each $n \in \mathbb N$, we have \begin{align} 
(\mathcal U\mathcal V)^n &= \sum_{k=1}^{n}  b^{n-k} \, S(n,k) \mathcal U(\mathcal U+a)(\mathcal U+2a) \cdots (\mathcal U+(k-1)a)  \mathcal V^k \notag \\
 &= \sum_{k=1}^{n} b^{n-k}  S(n,k)\,(\mathcal U \mid -a)_k\, \mathcal V^k  .\label{eqnUV}
\end{align}
\end{proposition}

\begin{remark}
Note that, in the existent literature, one would usually consider the normal ordering of $(\mathcal V\mathcal U)^n$ in which all operators $\mathcal V$ are to the left of the operators $\mathcal U$ (see e.g\ \cite[Section~8.5]{MansourSchork}   the references therein), while we are interested  in the opposite situation. The reader is advised to compare Proposition~\ref{UVtoN} with \cite{Mansour-1109}.
\end{remark}

\begin{proof}[Proof of Proposition~\ref{UVtoN}] First, we state that
\begin{equation} \label{VUlemma}
\mathcal V^n \mathcal U = (\mathcal U+na)  \mathcal V^n + nb  \mathcal V^{n-1}.
\end{equation}  
This formula follows immediately from \cite{Mansour-577}. Nevertheless, an interested reader can prove formula~\eqref{VUlemma} directly by induction. 

Now we prove \eqref{eqnUV} by induction. For $n=1$, \eqref{eqnUV} becomes the tautology $\mathcal U\mathcal V =\mathcal U\mathcal V$. Assume that \eqref{eqnUV} holds for $n$ and let us prove it for $n+1$. We have, by~\eqref{VUlemma}, 
\begin{align}
&(\mathcal U\mathcal V)^{n+1} = \sum_{k=1}^{n} b^{n-k}  S(n,k)\mathcal U(\mathcal U+a)(\mathcal U+2a) \cdots (\mathcal U+(k-1)a)  \mathcal V^k  \mathcal U\mathcal V   \notag \\
&\quad= \sum_{k=1}^{n}b^{n-k}  S(n,k) \ \mathcal U(\mathcal U+a)(\mathcal U+2a) \cdots (\mathcal U+(k-1)a)  \big[(\mathcal U+ka)  \mathcal V^k+ kb  \mathcal V^{k-1}\big]  \mathcal V \notag \\
&\quad= \sum_{k=1}^{n} b^{n-k}  S(n,k)  \mathcal U(\mathcal U+a)(\mathcal U+2a) \cdots (\mathcal U+ka)  \mathcal V^{k+1} \notag \\
&\qquad + \sum_{k=1}^{n} k  b^{n-k+1}  S(n,k) \mathcal U(\mathcal U+a)(\mathcal U+2a) \cdots (\mathcal U+(k-1)a)  \mathcal V^k . \label{tyhgb}
\end{align}
Setting $S(n,0) =S(n,n+1) = 0$, we continue \eqref{tyhgb} as follows:
\begin{align*}
&= \sum_{k=1}^{n+1}  S(n,k-1)  b^{n-k+1} \mathcal U(\mathcal U+a)(\mathcal U+2a) \cdots (\mathcal U+(k-1)a)  \mathcal V^k \\
&\quad + \sum_{k=1}^{n+1} S(n,k)  k  b^{n-k+1} \, \mathcal U(\mathcal U+a)(\mathcal U+2a) \cdots (\mathcal U+(k-1)a)  \mathcal V^k \\
&= \sum_{k=1}^{n+1} \big(S(n,k-1)+k  S(n,k)\big)  b^{n+1-k} \, \mathcal U(\mathcal U+a)(\mathcal U+2a) \cdots (\mathcal U+(k-1)a) \mathcal V^k \\
&= \sum_{k=1}^{n+1} b^{n+1-k}  S(n+1,k)  \mathcal U(\mathcal U+a)(\mathcal U+2a) \cdots (\mathcal U+(k-1)a)  \mathcal V^k ,
\end{align*}
where we used the well known recurrence formula $ S(n+1,k)=S(n,k-1)+k  S(n,k)$.
\end{proof}

Let now $\sigma > 0$, $\alpha,\beta\in\mathbb C\setminus\{0\}$. Define linear operators $\mathcal U$ and $\mathcal V$ in $\mathcal P (\mathbb C)$ by \eqref{vtdtstswt}.   It is straightforward to see that the operators $\mathcal U$, $\mathcal V$ generate a generalized Weyl algebra as discussed above with  $a=\beta$ and $b=\alpha-\beta$. Let $\mathcal R=\mathcal U\mathcal V$. 

Since $\mathcal V1=1$ and $\mathcal U = Z + \frac{\sigma}{\alpha}$, Proposition~\ref{UVtoN} immediately implies

\begin{corollary}\label{vctrsw645bvf}
We have
\begin{equation}
(\mathcal R^n1)(z)=\sum_{k=1}^{n} (\alpha-\beta)^{n-k}  S(n,k)\,(z+\sigma/\alpha \mid -\beta)_k .\label{cyrdei675}
\end{equation}
\end{corollary}

The following proposition explains a connection 
between the generalized Weyl algebra generated by $\mathcal U$ and $\mathcal V$ and an orthogonal Sheffer  sequence.

\begin{proposition}\label{f7iri7}
 Let $\sigma>0$ and $\alpha,\beta\in\mathbb C\setminus\{0\}$. Let $(p_n(z))_{n=0}^\infty$ be the monic polynomial  sequence satisfying by the recurrence formula 
\begin{equation}\label{vtyd6yed}
 zp_n(z)=p_{n+1}(z) +(\lambda n + \sigma/\alpha)  p_n(z) + (\sigma n + \eta n (n-1) ) p_{n-1} (z),\quad n\in\mathbb N_0,
\end{equation}
where $\lambda=\alpha+\beta$ and  $\eta=\alpha\beta$. In particular, for $\alpha\ge\beta>0$, we have $s_n(z)=p_n(z)$, and for $\Re(\alpha)\ge0$, $\Im(\alpha)>0$ and $\beta=\overline\alpha$, we have $s_n(z)=p_n(z+\frac\sigma\alpha)$. 
Define a linear bijective operator $\mathcal I$ in $\mathcal P(\mathbb C)$ by setting 
$\mathcal I(\cdot\mid\beta)_n=p_n$ for $n\in\mathbb N_0$. Then  $Z = \mathcal I \mathcal R {\mathcal I}^{-1}$.
\end{proposition}

\begin{proof} We have $ \mathcal R = \alpha Z D_\beta + Z + \sigma D_\beta + \frac\sigma\alpha$. Recall that $D_\beta$ is the lowering operator for the monic polynomial sequence $((z\mid\beta)_n)_{n=0}^\infty$.   Furthermore, it is easy to see that
$z(z\mid\beta)_n=(z\mid\beta)_{n+1}+n\beta(z\mid\beta)_n$. In view of the recurrence formula \eqref{vtyd6yed}, the statement easily follows. 
\end{proof}

As a special case of generalized Stirling numbers of Hsu and Shiue \cite{HsuShiue}, we define, for $0\le k\le n$ and $h,r\in\mathbb C$,  numbers $S(n,k;h,r)$ as the coefficients of the expansion $(z+r\mid h)_n=\sum_{k=0}^nS(n,k;h;r)(z\mid-h)_k$\,. 

Recall that the (unsigned) Lah numbers, $L(n,k)$, are defined as the coefficients of the expansion $(z)_n = \sum_{k=1}^{n} (-1) ^{n-k} L(n,k)  (z)^{(k)}$. Explicitly, $L(n,k) = \binom{n-1}{k-1} \frac{n!}{k!}$. Note that $L(n,k)= (-1) ^{n-k}S(n,k;1,0)=S(n,k;-1,0)$.

\begin{lemma} \label{yd6de}
We have $S(n,0;h, r) = (r \mid h)_n$ and for $k=1,\dots,n$, 
\begin{equation}\label{}
S(n,k;h,r)= \sum_{j=0}^{n-k} \binom{n}{j}  (-h)^{n-j-k}   L(n-j,k)  (r \mid h)_j. \end{equation}\end{lemma}

\begin{proof}
Since the $h$-derivative $D_h$ is the lowering operator for the monic polynomial sequence $((z\mid h)_n)_{n=0}^\infty$  and $(0\mid h)_n=0$ for all $n\in\mathbb N$, $((z\mid h)_n)_{n=0}^\infty$ is a polynomial sequence of binomial type, see e.g.\ \cite[4.3.3~Theorem]{11-Kung}. Hence, 
\begin{align*}
&( z+r \mid h)_n = \sum_{i=0}^{n} \binom{n}{i} (z \mid h)_i\, (r \mid h)_{n-i} 
= (r \mid h)_n + \sum_{i=1}^{n} \binom{n}{i}  h^i \left( \frac{z}{h} \right)_i (r \mid h)_{n-i} \\
&\quad= (r \mid h)_n + \sum_{i=1}^{n} \binom{n}{i}  h^i  (r \mid h)_{n-i}  \sum_{k=1}^{i} (-1)^{i-k} L(i,k)  \left( \frac{z}{h} \right)^{(k)} \\
&\quad= (r \mid h)_n + \sum_{i=1}^{n} \binom{n}{i} h^i  (r \mid h)_{n-i}  \sum_{k=1}^{i} (-1)^{i-k} L(i,k) \, h^{-k} (z \mid -h)_k \\
&\quad= (r \mid h)_n + \sum_{k=1}^{n} \left( \sum_{i=k}^{n} \binom{n}{i}  (-h)^{i-k} (r \mid h)_{n-i}\, L(i,k) \right) (z \mid -h)_k\\
&\quad= (r \mid h)_n + \sum_{k=1}^{n} \left( \sum_{j=0}^{n-k} \binom{n}{n-j}  (-h)^{n-j-k} (r \mid h)_{j}\, L(n-j,k) \right) (z \mid -h)_k. \qedhere
\end{align*}
\end{proof}

The following result can be of independent interest. 

\begin{theorem}\label{vctrs5u7}
Let $(p_n(z))_{n=0}^\infty$ be a monic polynomial  sequence as in Proposition~\ref{f7iri7}. We have
\begin{align}
z^n 
&=  \sum_{k=1}^{n} (\alpha - \beta)^{n-k} \, S(n,k) \,  (\sigma/\alpha \mid -\beta)_k \notag	\\
&\quad\quad + \sum_{i=1}^{n} \left( \sum_{k=i}^{n} (\alpha - \beta)^{n-k}  S(n,k)  S(k,i;-\beta,\sigma/\alpha) \right) p_i (z) \label{ctsa5y}
\end{align} 
and 
\begin{equation}
p_n (z) 
= (-\sigma/\alpha \mid \beta)_n + \sum_{i=1}^{n} \left( \sum_{k=i}^{n} S(n,k;\beta,-\sigma/\alpha)  (\alpha - \beta)^{k-i}  s(k,i)  \right) z^i. \label{hyt67oi} \end{equation}
\end{theorem} 

\begin{proof} By Corollary~\ref{vctrsw645bvf}, we have 
\begin{align}
(\mathcal R^n1)(z)&= \sum_{k=1}^{n} (\alpha - \beta)^{n-k}  S(n,k)  \sum_{i=0}^{k} S(k,i;-\beta, \sigma/\alpha) (z \mid \beta)_i \notag \\
&= \sum_{k=1}^{n} (\alpha - \beta)^{n-k} S(n,k) (\sigma/\alpha \mid -\beta)_k \notag \\
&\quad\quad + \sum_{i=1}^{n} \sum_{k=i}^{n} (\alpha - \beta)^{n-k}  S(n,k) \, S(k,i;-\beta, \sigma/\alpha) (z \mid \beta)_i .	\label{hj7uk}
\end{align} 
Applying the operator $\mathcal I$ to \eqref{hj7uk} and using Proposition~\ref{f7iri7}, we obtain \eqref{ctsa5y}.


Recall Corollary~\ref{vctrsw645bvf}. Note that there exists a unique monic polynomial sequence $(q_n (z))_{n=0}^{\infty}$ that satisfies $(\mathcal R^n 1)(z) = \sum_{k=1}^{n} (\alpha-\beta)^{n-k}  S(n,k) q_k (z)$ for $n\in {\mathbb N}_0$  and $q_n (z) = (z+\sigma/\alpha \mid -\beta)_n$. 

Define the monic polynomial sequence $(\tilde q_n(z))_{n=0}^{\infty}$ by
$$\tilde q_n(z)= \sum_{k=1}^{n} (\alpha-\beta)^{n-k}  s(n,k) \, (\mathcal R^n 1)(z).$$
We state that $q_n(z)=\tilde q_n(z)$, i.e., 
\begin{equation} \label{8709-jlf}
	 (z+\sigma/\alpha \mid -\beta)_n = \sum_{k=1}^{n} (\alpha-\beta)^{n-k}  s(n,k) \, (\mathcal R^n 1)(z).
						  \end{equation}
Indeed, using formula \eqref{65yhu7}, we have \begin{align*}
	&\sum_{k=1}^{n} (\alpha-\beta)^{n-k}  S(n,k) \tilde q_k (z)=\sum_{k=1}^{n} (\alpha-\beta)^{n-k}  S(n,k)  \sum_{i=1}^{k} (\alpha-\beta)^{k-i}  s(k,i)  (\mathcal R^i 1)(z) 	\\
		&\quad = \sum_{i=1}^{n} \left( \sum_{k=1}^{n}    S(n,k)  s(k,i) \right) (\alpha-\beta)^{n-i} (\mathcal R^i 1)(z)  = (\mathcal R^n 1)(z), 
											  \end{align*}
which proves \eqref{8709-jlf}. 

By Lemma~\ref{yd6de}, \eqref{8709-jlf} and the definition of the generalized Stirling numbers\linebreak $S(n,k;\beta, -\sigma/\alpha)$, we have \begin{align}
	(z \mid \beta)_n 
	&= \sum_{k=0}^n S(n,k; \beta, -\sigma/\alpha)(z+ \sigma/\alpha \mid -\beta)_k \notag	\\
	&= S(n,0; \beta, -\sigma/\alpha) + \sum_{k=1}^n S(n,k; \beta, -\sigma/\alpha) \left( \sum_{i=1}^{k} (\alpha-\beta)^{k-i}  s(k,i) (\mathcal R^i 1)(z) \right)	\notag	\\
	&= (-\sigma/\alpha\mid\beta)_n \, + \sum_{i=1}^{n} \left( \sum_{k=i}^{n} S(n,k; \beta, -\sigma/\alpha)  (\alpha-\beta)^{k-i}  s(k,i) \right) (\mathcal R^i 1)(z). \label{343Bm90op}
						  \end{align}
Applying $\mathcal I$ to \eqref{343Bm90op} and using Proposition~\ref{f7iri7}, we obtain \eqref{hyt67oi}. 	\qedhere

\end{proof}

The corollary below follows immediately from formula \eqref{ctsa5y}.

\begin{corollary} 	\label{654POI67}
	Let $(p_n (z))_{n=0}^\infty$ be a monic polynomial sequence as in Proposition~\ref{f7iri7}. Let $\Phi : \mathcal P (\mathbb C) \to \mathbb C$ be a linear functional defined by $\Phi(1) = 1$ and $\Phi(p_n) = 0$ for all $n \in \mathbb N$. Then $\Phi(z^n) = \sum_{k=1}^n (\alpha-\beta)^{n-k} \, S(n,k) \, (\sigma/\alpha \mid -\beta)_k$.	
	In particular, for $\alpha\ge\beta> 0$, 
	\begin{equation}\label{cfgdyjk} 
	\int_{\mathbb R} x^n \, \mu_{\alpha, \beta, \sigma}(dx) = \sum_{k=1}^n (\alpha-\beta)^{n-k} S(n,k) (\sigma/\alpha \mid -\beta)_k \end{equation}
	and for $\Re(\alpha)\ge0$, $\Im(\alpha)>0$, $\beta=\overline\alpha$,  
		\begin{equation}\label{vcdtre54w5uw5xs}
		 \int_{\mathbb R} (x + \sigma/\alpha)^n \, \mu_{\alpha, \beta, \sigma}(dx) = \sum_{k=1}^n (\alpha-\beta)^{n-k}  S(n,k)  (\sigma/\alpha \mid -\beta)_k . 
\end{equation}

\end{corollary}

\begin{center}
{\bf Acknowledgements} 
\end{center}

\noindent We are grateful to the anonymous referee for their valuable  suggestions, and for bringing papers \cite{Feinsilver2, Asai1,Asai2} to our attention.  
C.K. was financially supported  by the Doctoral Training Program (DTP), EPSRC, UKRI which co-operated with Faculty of Science and Engineering, Swansea University, the project reference  2602423, related to EP/T517987/1.


\begin{thebibliography}{99}


\bibitem{AccardiFranzSkeide} Accardi, L., Franz, U., and Skeide, M.,  ``Renormalized squares of white noise and other non-Gaussian noises as L\'evy processes on real Lie algebras,'' {\it Comm. Math. Phys.} {\bf 228}(1), 123--150 (2002). 

\bibitem{1-Accardi} Accardi, L.\ and Skeide, M., ``On the relation of the square of white noise and the finite difference algebra,'' {\it Infin. Dimens. Anal. Quantum Probab. Relat. Top.} {\bf 3}(1),  185--189 (2000). 


\bibitem{La-Othmane-23} Ali, S.T.\ and Ismail, M.E.H., ``Some orthogonal polynomials arising from coherent states,'' {\it J. Phys.\ A} {\bf 45}(12), 125203, 16 pp.\ (2012). 

\bibitem{Alpay-8} Alpay, D., J\o rgensen, P., Seager, R., and Volok, D., ``On discrete analytic functions: products, rational functions and reproducing kernels,'' {\it J. Appl. Math. Comput.} {\bf 41}(1--2), 393--426 (2013).

\bibitem{Alpay} Alpay, D.\ and Porat, M.,
``Generalized Fock spaces and the Stirling numbers,'' {\it J. Math. Phys.}  {\bf 59}(6), 063509, 12 pp.\ (2018).

\bibitem{Asai1} Asai, N., ``Hilbert space of analytic functions associated with the modified Bessel function and related orthogonal polynomials,'' {\it Infin. Dimens. Anal. Quantum Probab. Relat. Top.} {\bf 8}, 
505--514 (2005).

\bibitem{Asai2} Asai, N., ``Hilbert space of analytic functions associated with a rotation invariant measure,'' In  {\it Quantum Probability and Related Topics, QP–PQ: Quantum Probab. White Noise Anal., Vol.  23} (World Sci. Publ., Hackensack, NJ, 2008), pp.~49--62.

\bibitem{Segal-Bargmann} Asai, N., Kubo, I., and Kuo, H.-H.,
``Segal--Bargmann transforms of one-mode interacting Fock spaces associated with Gaussian and Poisson measures,'' {\it Proc. Amer. Math. Soc.} {\bf 131}(3), 815--823 (2003).

\bibitem{s-Bargmann} Bargmann, V., ``On a Hilbert space of analytic functions and an associated integral transform,'' {\it Comm. Pure Appl. Math.}\ {\bf 14}, 187--214 (1961).

\bibitem{BG} Barut, A.O. and Girardello, L., ``New ``coherent'' states associated with non-compact groups,'' {\it Comm. Math. Phys.} {\bf 21}, 41--55 (1971). 

\bibitem{Boukas1} Boukas, A., {\it Quantum Stochastic Analysis: A Non-Brownian Case, PhD thesis}  (Southern Illinois University, 1988).

\bibitem{Boukas2} Boukas, A., ``An example of a quantum exponential process,''  {\it Monatsh. Math.} {\bf 112}(3), 209--215  (1991).



\bibitem{Derezinski} Derezi\'nski, J. and G\'erard, C., {\it Mathematics of Quantization and Quantum Fields}  (Cambridge University Press, Cambridge, 2013). 

\bibitem{Feinsilver} Feinsilver, P., ``Discrete analogues of the Heisenberg--Weyl algebra,'' {\it Monatsh. Math.} {\bf 104}(2),  89--108 (1987).

\bibitem{Feinsilver2}  Feinsilver, P., ``Lie algebras and recurrence relations. I,'' {\it Acta Appl. Math.} {\bf 13}(3), 291--333 (1988). 

\bibitem{Sheffer} Finkelshtein, D., Kondratiev, Y., Lytvynov, E., Oliveira, M.-J., and Streit, L.,  ``Sheffer homeomorphisms of spaces of entire functions in infinite dimensional analysis,'' 
{\it J. Math. Anal. Appl.} {\bf 479}(1), 162--184 (2019).


\bibitem{Gazeau} Gazeau, J.P.,  {\it Coherent States in Quantum Physics} (Wiley, Weinheim, 2009).


\bibitem{GazeauKlauder} Gazeau, J.P. and Klauder, J.R., ``Coherent states for systems with discrete and continuous spectrum,'' {\it J. Phys. A} {\bf 32}(1), 123--132 (1999). 



\bibitem{Grabiner} Grabiner, S., ``Convergent expansions and bounded operators in the umbral calculus,'' 
{\it Adv. in Math.} {\bf 72}(1),  132--167 (1988). 

\bibitem{GrossMalliavin}Gross, L. and  Malliavin, P., 
``Hall's transform and the Segal-Bargmann map,''.In {\it It\^o's Stochastic Calculus and Probability Theory}, pp.~73--116 (Springer, Tokyo, 1996).

\bibitem{HsuShiue} Hsu, L.C. and Shiue, P.J.-S., ``A unified approach to generalized Stirling numbers,'' {\it Adv. in Appl. Math.} {\bf 20}(3), 366--384 (1998).



\bibitem{Mansour-577}  Irving, R.S., ``Prime ideals of Ore extensions over commutative rings. II,''  {\it J. Algebra} {\bf 58}(2), 399--423  (1979).

 
  
  
\bibitem{Karp1} Karp, D., ``Holomorphic spaces related to orthogonal polynomials and analytic continuation of functions,'' In {\it  Analytic extension formulas and their applications (Fukuoka, 1999/Kyoto, 2000}, pp.~169--187 (Kluwer Acad. Publ., Dordrecht, 2001).  

\bibitem{Karp2} Karp, D., ``Square summability with geometric weight for classical orthogonal expansions,'' In  {\it Advances in Analysis}, pp.~407--421 (World Sci. Publ., Hackensack, NJ, 2005).


  \bibitem{10-Katriel} Katriel, J., ``Combinatorial aspects of boson algebra,'' {\it Lett. Nuovo Cimento} {\bf 10}(13), 565--567 (1974).
  
  \bibitem{11-Kung} Kung, J.P.S., Rota, G.-C., and Yan, C.H., {\it Combinatorics: The Rota Way}  (Cambridge University Press, Cambridge, 2009). 
  
  \bibitem{Gamma-Lebedev} Lebedev, N.N.,  {\it Special Functions and their Applications. Revised English Edition} (Prentice-Hall Inc., Englewood Cliffs, 1965).
  
  \bibitem{Lytvynov}  Lytvynov, E., ``Polynomials of Meixner's type in infinite dimensions---Jacobi fields and orthogonality measures,''  {\it J. Funct. Anal.} {\bf 200}(1),  118--149 (2003). 
  
  \bibitem{MansourSchork} Mansour, T. and Schork, M., {\it Commutation Relations, Normal Ordering, and Stirling Numbers} (CRC Press, Boca Raton, 2016).
  
  \bibitem{13-Meixner} Meixner, J., ``Orthogonale Polynomsysteme mit einer besonderen Gestalt der erzeugenden Funktion, {\it J. London Math. Soc.} {\bf 9}(1), 6--13 (1934).
  
\bibitem{Meyer} Meyer, P.-A., {\it Quantum Probability for Probabilists} (Springer,  Berlin, 1993). 

\bibitem{Obata} Obata, N., {\it White Noise Calculus and Fock Space} (Springer,  Berlin, 1994).

  
  \bibitem{Parthasarathy} Parthasarathy, K.R., {\it An Introduction to Quantum Stochastic Calculus} (Birkh\"aser, Basel, 1992).

\bibitem{Perelomov} Perelomov, A., {\it Generalized Coherent States and their Applications} (Springer,  Berlin, 1986). 
  
  \bibitem{14-Quaintance} Quaintance, J.and Gould, H.W., {\it Combinatorial Identities for Stirling Numbers} (World Scientific, Singapore, 2016).
  
\bibitem{ReedSimon}  Reed M. and  Simon, B., {\it Methods of Modern Mathematical Physics II. Fourier Analysis, Self-Adjointness} (Academic Press, Boston, 1975).

\bibitem{Sivakumar} Sivakumar, S., ``Studies on nonlinear coherent states,''  {\it J. Opt. B Quantum Semiclass. Opt.} {\bf 2}(6), R61--R75 (2000).

\bibitem{Mansour-1109} Viskov, O.V., ``On a theorem of R. A. Sack for shift operators,''  {\it Dokl. Akad. Nauk} {\bf 340}(4), 463--466 (1995).

\end{thebibliography}
\end{document}